\documentclass[11pt,oneside,letterpaper]{article}
\usepackage{amssymb}
\usepackage{amsmath}
\usepackage{amsthm}
\usepackage[dvips]{graphicx}
\usepackage{setspace}
\usepackage{amsfonts}
\usepackage{fancyhdr}
\usepackage{xcolor}
\usepackage{graphicx}
\usepackage{rotating}
\usepackage{comment}
\usepackage{color}
\usepackage{cite}
\usepackage{braket}
\usepackage{caption}
\usepackage[utf8]{inputenc}
\usepackage{mathtools}
\newtheorem{lemma}{Lemma}

\definecolor{darkgreen}{rgb}{0,0.5,0}
\definecolor{darkblue}{rgb}{0,0,0.6}
\definecolor{purple}{rgb}{0.4,.2,0.7}
\newcommand{\p}{\partial}

\newcommand{\be}{\begin{equation}}
\newcommand{\ee}{\end{equation}}

\usepackage[colorlinks=true,citecolor=darkgreen,linkcolor=black,urlcolor=purple]{hyperref}

\usepackage{pdfsync}

\makeatletter
\newcommand*{\defeq}{\mathrel{\rlap{%
                     \raisebox{0.3ex}{$\m@th\cdot$}}%
                     \raisebox{-0.3ex}{$\m@th\cdot$}}%
                     =} 
\makeatother

\def\be{\begin{eqnarray}}
\def\ee{\end{eqnarray}}

\newcommand{\bea}{\begin{eqnarray}}
\newcommand{\eea}{\end{eqnarray}}
\def\ben{\begin{equation}}
\def\een{\end{equation}}

    \let\p=\phi \let\r=v

\def\be{\begin{equation}}
\def\ee{\end{equation}}
\def\ba{\begin{array}}
\def\ea{\end{array}}

\def\ba#1\ea{\begin{align}#1\end{align}}
\def\bs#1\es{\begin{split}#1\end{split}}

\renewcommand{\p}{\partial}

\newcommand{\bz}{\bar{z}}

\newcommand{\bq}{\bar{q}}
\newcommand{\btau}{\bar{\tau}}

\allowdisplaybreaks  

\numberwithin{equation}{section}

\def \be {\begin{equation}}
\def \ee {\end{equation}}
	
\def \JM#1 {{\color{blue}  JM: #1 }}
\def \AAl#1 {{\color{red}  AA: #1 }}

\begin{document}
\onehalfspacing

\begin{center}

~
\vskip5mm

{\LARGE  {
Averaging over moduli in deformed WZW models
\\
\ \\
}}

Junkai Dong, Thomas Hartman, and Yikun Jiang

\vskip5mm
{\it Department of Physics, Cornell University, Ithaca, New York, USA
} 

\vskip5mm

{\tt jd722@cornell.edu, hartman@cornell.edu, yj366@cornell.edu }

\end{center}

\vspace{4mm}

\begin{abstract}
\noindent
WZW models live on a moduli space parameterized by current-current deformations. The moduli space defines an ensemble of conformal field theories, which generically have $N$ abelian conserved currents and central charge $c>N$. We calculate the average partition function and show that it can be interpreted as a sum over 3-manifolds. This suggests that the ensemble-averaged theory has a holographic dual, generalizing recent results on Narain CFTs. The bulk theory, at the perturbative level, is identified as $U(1)^{2N}$ Chern-Simons theory coupled to additional matter fields. From a mathematical perspective, our principal result is a Siegel-Weil formula for the characters of an affine Lie algebra.

 \end{abstract}

\pagebreak
\pagestyle{plain}

\setcounter{tocdepth}{2}
{}
\vfill

\ \vspace{-2cm}
\renewcommand{\baselinestretch}{1}\small
\tableofcontents
\renewcommand{\baselinestretch}{1.15}\normalsize

\newpage
\section{Introduction}

There are indications that the sum over topologies in the gravitational path integral is related to an ensemble average over microscopic theories \cite{Coleman:1988cy, Giddings:1988cx, 1806.06840, 1903.11115, 1910.10311, Marolf:2020xie, 2006.08648}. The clearest example is the duality between two-dimensional JT gravity and random matrix theory, which suggests a considerable extension to the AdS/CFT dictionary \cite{1903.11115}. 

There are different points of view on the recent developments resurrecting this  old idea. The first is that the implications of the gravitational path integral should be trusted, and non-factorizing contributions force us to consider an ensemble average (or, equivalently, superselection sectors). The second is that the sum over topologies is an uncontrolled artifact of the low-energy description, and the true microscopic theory is ordinary quantum mechanics without disorder. There is also a middle path, natural to an effective field theorist: The microscopic theory is ordinary quantum mechanics, but the sum over topologies accurately computes certain self-averaging observables. With this interpretation, we should \textit{either} include the sum over topologies \textit{or} UV microstates, not both, as the higher topologies encode some coarse-grained properties of the microstates.
There is direct evidence for this ``topological dualism" in string theory \cite{2008.07533, 2102.12355, 2004.06738, hep-th/0401024,0705.2768} and the SYK model \cite{2103.16754, 2105.08207}, and circumstantial evidence more generally, including the successful application of wormholes to the Page curve \cite{1911.12333, 1911.11977}.  

In more than two bulk dimensions, it is far from clear how to define an ensemble average of microscopic theories. Progress has been made in AdS$_3$ gravity in certain limits, e.g.~\cite{1912.07654, 2006.11317, 2006.08648, 2007.15653}. An alternative is to simplify the problem by adding more symmetries. Afkhami-Jeddi, Cohn, Hartman, and Tajdini \cite{2006.04839} and Maloney and Witten \cite{2006.04855} found that one simple ensemble of 2d CFTs --- the Narain ensemble of toroidal compactifications of $N$ free bosons --- has a holographic interpretation. After averaging over the Narain moduli, the partition function of the theory on any Riemann surface, or collection of disconnected Riemann surfaces, can be reinterpreted as a sum over three-dimensional topologies. The bulk theory is not ordinary gravity, because it has a large number of light states, but the same is true of other models where we can perform the sum over topologies, including tensionless string theory \cite{1803.04423,1911.00378} and the SYK model \cite{cond-mat/9212030, Kitaev1}. 

At present, there is no fully satisfactory, non-perturbative definition of the bulk theory dual to the Narain ensemble. It is related to $U(1)^{2N}$ Chern-Simons theory, but there is a restriction on the allowed gauge configurations and the sum over topologies appears \textit{ad hoc}. Therefore, conservatively, the bulk can be interpreted as an effective theory that is UV-completed by an individual CFT in the ensemble, and the question is what aspects of the UV are encoded in the IR sum over topologies. In theories of 2d gravity with a hard cutoff on the dilaton, such as those obtained by reducing from higher dimensions, the JT/Random matrix duality \cite{1903.11115} should be interpreted in the same spirit.

It would be useful to have more examples of averaged CFTs in order to explore these issues further. The original Narain duality has been extended in several directions recently, to include chemical potentials \cite{2102.12509}, orbifolds of free bosons \cite{2103.15826}, and more general quadratic forms \cite{2104.14710} (see also \cite{2006.08216,2009.01244,2012.15830,2102.03136,2104.10178} for related progress). In this paper, we generalize the Narain duality to an ensemble of CFTs defined by the moduli space of the $SU(N+1)_k$ WZW model deformed by exactly marginal current-current operators. The action is
\be\label{introS}
S_{\rm dWZW}  = S_{\rm WZW} + \sum_{i,j=1}^N \lambda_{ij} \int d^2 z J^i(z) \bar{J}^j(\bar{z}) \ ,
\ee
where $J^i$ are the Cartan currents. These deformed WZW models are exactly solvable, as described by F\"orste and Roggenkamp \cite{hep-th/0304234}. The moduli space is easy to describe. WZW models can be realized as the orbifolds 
\be\label{introCoset}
SU(N+1)_k = \left( \frac{SU(N+1)_k}{U(1)^N} \times T_\Lambda \right) / \, \Gamma_k \ ,
\ee
where the first factor is known as a generalized parafermion theory \cite{Gepner:1987sm, Gepner:1986hr},  $T_\Lambda$ is a Narain CFT, and $\Gamma_k$ is a discrete group. The current-current deformation acts only on the $T_\Lambda$ theory. Therefore the moduli space of deformed WZW models is locally identical to the Narain moduli space (but there are global differences related to dualities). 
The ensemble average is defined by integrating over the couplings $\lambda_{ij}$ with the Zamolodchikov measure.
We will see that the methods used to analyze the averaged Narain theory in \cite{2006.04839,2006.04855,2102.12509,2103.15826,2104.14710}, based on the Siegel-Weil theorem, can be extended to study this more general ensemble. The WZW model at level $k=1$ is a theory of $N$ free bosons, so we recover the Narain duality in this special case.

 Our main result is a simple formula for the averaged torus partition function, 
\be\label{introZ}
\langle Z_{\rm dWZW}(\tau) \rangle = \sum_{\gamma \in \Gamma_{\infty}\backslash PSL(2,\mathbb{Z})} \sum_{\lambda} |c_0^\lambda(\gamma \tau)|^2
\ee
where `dWZW' stands for `deformed WZW', $\lambda$ runs over primaries of $SU(N+1)_k$, and $c^\lambda_\mu$ is the string function of the affine Lie algebra. The sum over $PSL(2,\mathbb{Z})$ in this formula suggests a holographic interpretation as a sum over 3-manifolds. We will not have much to say about the bulk theory, but based on the partition function we conclude that perturbatively, it should be a theory of $U(1)^{2N}$ Chern-Simons fields coupled to topological matter dual to the parafermions.   The $U(1)$ sector and the parafermion sector are coupled only in the sense that they are correlated by the sum over manifolds. The non-perturbative definition of the bulk theory is incomplete, similar to the Narain duality.

It is remarkable that the final answer \eqref{introZ} is so simple. This is a consequence of various cancellations. Note that the parafermions do not directly participate in the averaging procedure, but nonetheless their characters, which are proportional to $c^\lambda_{\mu}$, provide the seed for the sum over manifolds.

This duality has some nice features beyond the original Narain duality. CFTs in the ensemble generically have $U(1)^N_{\rm left} \times U(1)^N_{\rm right}$ symmetry, like the Narain duality, but a larger central charge, $c>N$. There are non-conserved primary operators coming from the parafermion sector that survive the averaging procedure, so the bulk theory has matter fields. This allows for the study of averaged correlation functions, although we will not pursue this here. (Correlators were also studied recently in an orbifolded version of the Narain duality and found to have an interesting correspondence to rational tangles \cite{2103.15826}.) Also, since there is a coupling constant $1/k$, the theory has an 't Hooft-like limit $N,k \to \infty$. This limit resembles earlier versions of higher-spin AdS/CFT \cite{1011.2986, 1011.5900}. Finally the WZW models appear in various contexts in string theory, and we view our result as a step toward averaging more general sigma models. Any sigma model with enough symmetry can be recast as an abelian coset of a WZW model \cite{hep-th/9302033}, so at least these cases should be tractable.

An intriguing aspect of averaged holographic dualities is the interplay with global symmetries and the swampland \cite{2004.06738,2010.10539,2011.06005,2011.09444,2012.07875}. The arguments against global symmetries in quantum gravity \cite{Misner:1957mt,hep-th/0304042,1011.5120,1810.05337,2010.10539} do not apply to ensemble-averaged theories, because factorization fails, and black hole evaporation is not unitary in the usual sense. Thus global symmetries can arise after averaging.

The pattern in the Narain duality and deformed WZW duality is a bit different from the case of JT/Random matrix duality. At the perturbative level, the bulk theory has an apparent global symmetry $O(N,N)$ rotating the Chern-Simons fields. The moduli space that defines the CFT ensemble is a quotient of the same group, $O(N,N)$. Therefore by summing over topologies, we end up averaging over the apparent global symmetry group of the bulk theory. It is an interesting question whether this pattern applies more generally to theories of quantum gravity with apparent global symmetries. Perhaps whenever a bulk theory appears to have a global symmetry $G$, the sum over topologies induces an average over $G$. This is related to observations in \cite{2010.10539,2011.06005,2011.09444} that symmetries of the ensemble can appear as global symmetries in the bulk.

The outline of the paper is as follows. In section \ref{s:wzw} we review current-current deformations of WZW models, and set up the calculation of the average. Unlike in \cite{2006.04839,2006.04855}, we must consider twisted sectors of the Narain factor. Therefore in section \ref{s:theta} we calculate the average partition function of a toroidal CFT with twisted boundary conditions (this result stands independently and could be useful in other ensembles). In section \ref{s:averagewzw} we return to the WZW model and  derive the average partition function quoted above in \eqref{introZ}. We also consider the averaged spectrum. Section \ref{s:discussion} is a discussion of the holographic dual and some open questions.

\section{Current-current deformations of WZW models}\label{s:wzw}

\subsection{WZW model as orbifold}

Wess-Zumino-Witten (WZW) models are 2d rational CFTs associated to affine Lie groups \cite{Wess:1971yu, Witten:1983tw, Witten:1983ar , Novikov:1982ei, DiFrancesco:1997nk}. We will restrict to the case of $SU(N+1)_k$ with a diagonal spectrum, but it should be straightforward to extend the analysis to other compact (and possibly non-compact) groups. The central charge is
\be
c = \frac{k N(N+2) }{k+N+1} \ .
\ee
The level $k$ is a positive integer, and we assume $N>2$ because as we will see, otherwise the average is ill defined. Denote the root lattice of $SU(N+1)$ by $Q$ and the weight lattice by $P$. Let $e_i$ be the simple roots and $e^i$ the fundamental weights, with $e^i \cdot e_j = \delta^i_j$. Our convention is to suppress indices on lattice vectors but write other indices explicitly.

The torus partition function is a finite sum of affine characters,
\be \label{WZW}
Z_{\rm WZW}(\tau) = \sum_{\lambda \in P_k^+} |\chi_{\lambda}(\tau)|^2 \ ,
\ee
where $P_k^+ \subset P$ are the (finite parts of) integrable highest weights at level $k$, which correspond to Young tableaux with at most $k$ boxes in the first row.
Representations of the affine Lie algebra $\hat{\mathfrak{g}}_k$ decompose into representations of the coset $\hat{\mathfrak{g}}_k/\hat{\mathfrak{h}}$, where $\mathfrak{h}$ is the Cartan subalgebra of $\mathfrak{g}$. This underlies the orbifold construction in \eqref{introCoset}. Primaries in the WZW model decompose into the parafermions and toroidal bosons as
\be
\mathcal{V}_{\lambda} \cong \bigoplus_{\mu \in P/kQ}  \mathcal{V}^{PF}_{\lambda, \mu} \otimes \mathcal{V}^{\Lambda}_{\mu} \ .
\ee
Thus the characters can be expanded in theta functions,
\be\label{stringdecomp}
\chi_{\lambda}(\tau) = \sum_{\mu \in P/kQ} c^\lambda_\mu(\tau) \theta_{\mu}(\tau)
\ee
where the coefficients $c^\lambda_\mu$ are called string functions, and 
\be\label{smalltheta}
\theta_{\mu}(\tau) = \sum_{\delta \in Q} q^{\frac{1}{2k}(\mu+k\delta)^2}
\ee
with $q = e^{2\pi i \tau}$. The string functions are complicated objects, but they have relatively nice transformation laws under $SL(2,\mathbb{Z})$ \cite{Kac:1990gs}. Up to a factor of $\eta(\tau)^N$, the string functions are the branching functions for the $\hat{\mathfrak{g}}_k/\hat{\mathfrak{h}}$ coset.

An orbifold representation of WZW models as in \eqref{introCoset} was described by F\"orste and Roggenkamp \cite{hep-th/0304234}, following earlier work in \cite{hep-th/9302033, Yang:1988bi, Gepner:1987sm}. We will use a slightly different construction (see appendix \ref{app:orbifold} for a comparison). 
The orbifold group is
\be
\Gamma_k \cong Q/kQ \ .
\ee
To describe its action on the Hilbert space, let $G^0_{ij} = e_i \cdot e_j$ be the Cartan matrix and define the antisymmetric matrix $B^0_{ij} = G^0_{ij}$ for $i<j$, $B^0_{ij} = -G^0_{ij}$ for $i>j$. Denote 
\be \label{B0}
B^0 = e^i B^0_{ij} e^j \ .
\ee
Below we will use the shorthand $B^0 \beta := e^i B^0_{ij} ( e^j \cdot \beta)$. The symmetry $\gamma \in \Gamma_k$ acts on the (left $\otimes$ right) parafermion Hilbert space as
\begin{align}\label{gammaaction}
\mathcal{V}^{PF}_{\lambda, \mu} \otimes \overline{\mathcal{V}}^{PF}_{\lambda, \bar{\mu}}  &\to e^{i \pi \gamma \cdot(\mu + \bar{\mu}) /k- i \pi \gamma \cdot B^0 \cdot (\mu - \bar{\mu})/k } \mathcal{V}^{PF}_{\lambda, \mu} \otimes \overline{\mathcal{V}}^{PF}_{\lambda, \bar{\mu}}
\end{align}
and with the opposite phase on the bosons. 
 The partition functions of the two factors, twisted by group elements $\alpha, \beta \in \Gamma_k$, are
\begin{align}\label{twistedZPF}
Z^{PF}_{\alpha,\beta}(\tau) &= 
\frac{1}{N+1} |\eta(\tau)|^{2N} \sum_{\lambda \in P_k^+} \sum_{\mu \in P/kQ} e^{i\pi \alpha \cdot(2\mu-\beta-B^0 \beta)/k}c_\mu^\lambda(\tau) \bar{c}^\lambda_{\mu-\beta}(\btau) \\*
Z_{\alpha,\beta}^{\Lambda}(\tau) &= 
|\eta(\tau)|^{-2N} \sum_{\mu \in P/kQ} e^{-i \pi \alpha \cdot(2\mu-\beta-B^0\beta)/k} \theta_{\mu}(\tau) \bar{\theta}_{\mu-\beta}(\btau) \label{twistedZB}
\end{align}
for the parafermions and the free bosons, respectively. These combine to reproduce the WZW model partition function \eqref{WZW} by summing over twisted sectors and dividing by the order of $\Gamma_k$,
\begin{align}\label{zOrb}
Z_{\rm WZW}(\tau)  = k^{-N} \sum_{\alpha,\beta \in Q/kQ} Z_{\alpha,\beta}^{PF}(\tau) Z^\Lambda_{\alpha,\beta}(\tau) \ .
\end{align}
For more details and some useful string function identities, as well as a comparison to F\"orste and Roggenkamp, see appendix \ref{app:orbifold}.

The free boson sector is a toroidal CFT, so we can rewrite $Z_{\alpha,\beta}^{\Lambda}$ as a sum over a Narain lattice. By plugging \eqref{smalltheta} into \eqref{twistedZB} we see that the lattice is
\be \label{latticewzw}
\Lambda =  \frac{1}{\sqrt{k}} \{ (x, x') \in P \times P \, | \ x-x' \in k Q \} \ .
\ee
This is an even self-dual lattice in $\mathbb{R}^{N,N}$, with the inner product
\be
(\lambda_+, \lambda_-) \cdot (\mu_+, \mu_-) = \lambda_+ \cdot \mu_+ - \lambda_- \cdot \mu_- \ ,
\ee
where each entry is a vector in $\mathbb{R}^N$. Define the Siegel-Narain theta function
\begin{align}\label{thetaLambda}
\Theta_{\Lambda}(a,b,\tau) &= \sum_{\lambda \in \Lambda} q^{\frac{1}{2}(\lambda+b)_+^2} \bq^{\frac{1}{2}(\lambda+b)_-^2} e^{-2\pi i a\cdot(\lambda + \frac{1}{2}b )} \ .
\end{align}
The twisted boson partition function given in \eqref{twistedZB} is 
\begin{align}\label{zTheta}
Z_{\alpha,\beta}^{\Lambda}(\tau) &= |\eta(\tau)|^{-2N}  \Theta_{\Lambda}(w(\alpha), w(\beta), \tau) \ .
\end{align}
The twists $\alpha,\beta$ are vectors in $\mathbb{R}^N$ while $w(\alpha), w(\beta)$ are vectors in $\mathbb{R}^{N,N}$. In \eqref{zTheta} they are related by\footnote{Derivation: Starting from \eqref{thetaLambda}, write $\lambda = \frac{1}{\sqrt{k}}(\mu + k \delta, \mu + k \delta')  -\lambda_0$ where $\mu \in P/kQ$, $\delta, \delta' \in Q$, and the shift is by the lattice vector $\lambda_0 =  \frac{1}{\sqrt{k}}(p,p)$ with $p = \frac{1}{2}(1+B^0)\beta \in P$. Now plugging in the twists 
\eqref{twistRel} immediately gives \eqref{twistedZB}.}
\begin{align}\label{twistRel}
w(\alpha) &= \frac{1}{2\sqrt{k}}((1+B^0)\alpha, (-1+B^0)\alpha )  \ .
\end{align}
These twists are special in that they lie in the lattice $(\frac{1}{k}\Lambda)/\Lambda$ and satisfy
\be\label{absq}
w(\alpha)^2 = w(\beta)^2 = w(\alpha)\cdot w(\beta) = 0  
\ee
under the Narain inner product.
More details on the Narain lattice of the WZW model can be found in appendix \ref{app:narain}. 

\subsection{Current-current deformations}

The orbifold representation allows for a straightforward analysis of the moduli space of current-current deformations \cite{hep-th/0304234}. The conserved currents of the WZW model are $J^a(z)$ and $\bar{J}^a(\bz)$ with $a=1\dots \dim \mathfrak{g}$ running over the generators of the Lie algebra $\mathfrak{g}$. Currents in the Cartan subalgebra, which we denote $J^i(z)$ and $\bar{J}^i(\bz)$ with $i=1 \dots \mbox{rank\ }\mathfrak{g}$, are mutually commuting. Therefore for $SU(N+1)_k$ there are $N^2$ exactly marginal operators, $J^i(z) \bar{J}^j(\bz)$. The WZW model deformed by these operators as in \eqref{introS} is conformal for finite values of the couplings $\lambda_{ij}$. (Conformal invariance would be broken if we deformed by non-commuting currents.) The deformations modify the spectrum but leave the OPE coefficients unchanged \cite{hep-th/0304234, Chaudhuri:1988qb}.

In terms of the orbifold, the Cartan currents live in the Narain factor, so the deformation acts only on $T_{\Lambda}$. The deformed WZW models are therefore orbifolds \cite{hep-th/0304234}
\be
\mbox{dWZW} = \left( \frac{SU(N+1)_k}{U(1)^N} \times T_{O \Lambda} \right)/\Gamma_k \ ,
\ee
where $O \in O(N,N)$ deforms the Narain lattice $\Lambda$. 
The partition function is
\be
Z_{\rm dWZW}(\tau) = k^{-N} \sum_{\alpha,\beta \in Q/kQ} Z^{PF}_{\alpha,\beta}(\tau) Z^{O \Lambda}_{\alpha,\beta}(\tau) \ ,
\ee
where
\be\label{zOdef}
Z_{\alpha,\beta}^{O\Lambda}(\tau) = |\eta(\tau)|^{-2N}  \Theta_{O\Lambda}(O w(\alpha), Ow(\beta), \tau) \ .
\ee
Note that the deformation rotates the vectors $w(\alpha),w(\beta)$. This is required in order for $Z_{\alpha,\beta}^{O\Lambda}$ to stay in a fixed, finite-dimensional representation of $SL(2,\mathbb{Z})$ after the deformation, so that the full partition function $Z_{\rm dWZW}$ remains modular invariant. In other words, the coordinates of the vectors $w(\alpha),w(\beta)$ are held fixed under deformations, but the basis vectors depend on moduli.

This construction implies that the moduli space of deformed WZW models is locally identical to the Narain moduli space. Globally, the full moduli space is
\be
{\cal M} = O(N) \times O(N) \backslash O(N,N) / \mbox{Dualities} \ ,
\ee
where $O(N) \times O(N)$ are rotations acting individually on the two blocks of $\mathbb{R}^{N,N}$. For a Narain theory alone, the duality group would be the automorphism group of the original lattice, which is isomorphic to $O(N, N, \mathbb{Z})$. In the deformed WZW model there are further restrictions on the dualities because the OPE coefficients must be preserved, so in fact for $N,k>1$ the duality group is a finite subgroup of $O(N,N,\mathbb{Z})$ \cite{hep-th/9302033,hep-th/0304234}.\footnote{Curiously, this means that for each point in the Narain moduli space, there is an infinite set of inequivalent CFTs that all have the same partition function. Although the OPE coefficients are preserved by the deformation, the CFTs are inequivalent because the assignment of OPE coefficients to Virasoro primaries of given weight gets reshuffled by elements of $O(N,N,\mathbb{Z})$.}

\subsection{Setting up the average over moduli}
Under the Haar measure, the moduli space ${\cal M}$ has infinite volume for $k>1$, because the symmetric space $O(N) \times O(N) \backslash O(N,N)$ has infinite volume and the duality group is finite. This is in contrast to the Narain moduli space, which has finite volume for $N \geq 2$. These moduli spaces are locally identical, but in the deformed WZW model, the duality group is finite so ${\cal M}$ is generally much larger. 

To average over the moduli space of the deformed WZW model we must regulate this infinity. Fortunately, for the partition function there is a natural choice:  we simply average over the Narain moduli space instead, 
\be
{\cal M}' = O(N) \times O(N) \backslash O(N,N) / O(N,N,\mathbb{Z}) \ .
\ee
This is well defined because the partition function has the full $O(N,N, \mathbb{Z})$ invariance.
For $N\neq 1$ and $k \neq 1$, it divides the average over ${\cal M}$ by an infinite constant. Thus we define the ensemble average partition function by
\begin{align}
\langle Z_{\rm dWZW} (\tau) \rangle &= \int_{ {\cal M}' } d(\mbox{moduli})  \, Z_{\rm dWZW} (\tau) \ .
\end{align}
We emphasize that this is our choice of ensemble. Any observable can be averaged over this ensemble, although for general observables sensitive to the OPE coefficients, the result may depend on the choice of fundamental domain for $O(N,N,\mathbb{Z})$. For the partition function this not an issue.
Only the boson factor depends on moduli, so 
\begin{align}\label{avZsetup}
\langle Z_{\rm dWZW} (\tau) \rangle &= k^{-N}  |\eta(\tau)|^{-2N} \sum_{\alpha,\beta \in Q/kQ} Z^{PF}_{\alpha,\beta}(\tau) \big\langle \Theta_{\Lambda}(w(\alpha), w(\beta), \tau) \big\rangle \ .
\end{align}
We will carefully define and calculate $\langle \Theta_{\Lambda}\rangle$ in the next section and return to the WZW model in section \ref{s:averagewzw}.

\section{Averaging twisted theta functions}\label{s:theta}

Our first task is to compute the average over moduli space of a Siegel-Narain theta function, with nonzero twists $(a,b)$. The twisted sectors with $a=0$ were treated by Siegel, and interpreted holographically in recent work generalizing the Narain duality to orbifolds \cite{2103.15826} and general lattices \cite{2104.14710}. The main new element here is to allow for independent $a$ and $b$.

\subsection{Siegel-Narain theta functions}

Although we already defined the Siegel-Narain theta function in \eqref{thetaLambda}, this definition is inconvenient for averaging because the twist vectors depend on moduli. We will now introduce a more convenient notation similar to Siegel's.

Let $S_{\mu\nu}$ be a semi-integral quadratic form, \textit{i.e.}, a symmetric matrix with integers on the diagonal and half-integers off the diagonal. For any positive definite matrix $H_{\mu\nu}$ satisfying $H S^{-1} H = S$, define the theta function
\begin{align}\label{fHS}
\Theta_{H, S}(a^\mu, b^\mu, \tau) &= \sum_{n^\mu \in \mathbb{Z}^{2N}} e^{-2\pi  \tau_2 H_{\mu\nu} (n^\mu+b^\mu)(n^\nu + b^\nu) + 2\pi i \tau_1  S_{\mu\nu} (n^\mu+b^\mu)( n^\nu+b^\nu) - 4\pi i S_{\mu\nu} a^\mu (n^\nu + \tfrac{1}{2}b^\nu) } \ ,
\end{align}
where $a^\mu$ and $b^\mu$ are vectors in $\mathbb{R}^{2N}$ called characteristics or twists, and $\tau = \tau_1 + i \tau_2$ with $\tau_{1} \in \mathbb{R}$ and $\tau_2 \in \mathbb{R}_{+}$. We will assume $S_{\mu\nu}$ has signature $(N,N)$ with $N > 2$. 

Siegel developed a formalism to analyze general semi-integral quadratic forms, not necessarily associated to self-dual lattices. \cite{Siegel1,Siegel2,Siegel3} We are interested in the self-dual case relevant to Narain lattices (but we will end up using Siegel's more general result at an intermediate step). The quadratic form associated to Narain lattices is
\be\label{narainS}
T_{\mu\nu} = \frac{1}{2} \begin{pmatrix} 0 & \mathbf{1}_{N\times N} \\ \mathbf{1}_{N\times N} & 0 \end{pmatrix} \ ,
\ee
and the moduli of the Narain lattice are encoded in the block matrix $H$ as
\be
H_{\mu\nu} =\left(\begin{array}{cc} G^{-1}& \frac{1}{2} G^{-1} B\\ -\frac{1}{2} B G^{-1} & \frac{1}{4} (G-B G^{-1} B)\end{array}\right)
\ee
where $G_{ij}$ is the metric and $B_{ij}$ is the antisymmetric $B$-field. We will denote the inner product with the Narain quadratic form by
\be
\langle x, y \rangle = T_{\mu\nu} x^\mu y^\nu \ .
\ee
The theta function $\Theta_{H,T}$ is identical to the Siegel-Narain theta function written as a sum over lattice points in \eqref{thetaLambda},
\be
\Theta_{H,T}(a^\mu, b^\mu, \tau) = \Theta_{\Lambda}(a^\mu U_\mu, b^\mu U_\mu,\tau)
\ee
where the vectors $U_\mu$ are the basis vectors of the Narain lattice as in appendix \ref{app:narain}. The notation $\Theta_{H,T}$ is more convenient for doing the average because the coordinate $a^\mu$ is rational and independent of moduli, unlike the vector $a = a^\mu U_{\mu}$ as explained around \eqref{zOdef}.

To describe the modular transformations of $\Theta$, let
\be\label{defgamma}
\gamma = \begin{pmatrix} f & g \\ c & d \end{pmatrix} \in SL(2,\mathbb{Z}) \ .
\ee
Package the twists into a vector
\be
A^\mu = \begin{pmatrix} a^\mu \\ b^\mu \end{pmatrix} \ .
\ee
The Siegel-Narain theta function for an even self-dual lattice of signature $(N,N)$ transforms as
\be\label{modtheta}
\Theta_{H,T}(\gamma A^\mu, \gamma \tau) =  |c\tau+d|^N \Theta_{H,T}(A^\mu, \tau) \ ,
\ee
where $\gamma$ acts on $\tau$ by fractional linear transformations and on the vector $A^\mu$ by ordinary matrix multiplication,
\be
\gamma \tau = \frac{f \tau + g}{c\tau + d}  \ , \qquad
\gamma A^\mu = \begin{pmatrix} f a^\mu + g b^\mu \\ c a^\mu + d b^\mu \end{pmatrix} \ .
\ee
The formula \eqref{modtheta} is equivalent to the statement that the combination $\Theta / |\eta|^{2N}$ appearing in the partition function of free bosons is modular invariant, so long as we transform both $\tau$ and the twists.

\subsection{Calculating the average}
The definitions and transformation rules above hold for arbitrary twists $a$ and $b$. For generic twists, the action of modular transformations on $\Theta_{H,T}$ produces an infinite dimensional representation of $SL(2,\mathbb{Z})$. However we will be interested in special cases where this action truncates and the representation is finite-dimensional. This occurs when $a^\mu$, $b^\mu$ are rational, because if $b^{\mu}{}' = c a^\mu + d b^\mu$ is integral, then the shift by $b'$ in the theta sum is trivial. We therefore restrict to
\be
a^\mu \in \frac{1}{k} \mathbb{Z}^{2N} , \qquad b^\mu \in \frac{1}{k} \mathbb{Z}^{2N} \ ,
\ee
where $k$ is a positive integer (which will eventually be identified as the WZW level).

We would like to calculate the average over Narain moduli space,
\be\label{averageint}
\langle \Theta_{T}(a^\mu, b^\mu, \tau)\rangle := \int dH \, \Theta_{H,T}(a^\mu, b^\mu, \tau) \ .
\ee
The integral is over the moduli space $O(N)\times O(N) \backslash O(N,N)/O(N,N,\mathbb{Z})$, and $dH$ is the Haar measure normalized by the volume of the moduli space. We will not need the explicit formula for the measure but it can be found in \cite{Siegel1,Siegel2,Siegel3, 2006.04855}.

Siegel calculated the average for $a=0$. Fortunately turning on nonzero $a$ is a straightforward generalization, and the integral can be calculated by any of the various methods described in \cite{2006.04839} and \cite{2006.04855}, or (with somewhat more difficulty) by reducing it to a finite sum of Siegel averages. Whatever method we use, the basic idea is that $\langle \Theta_T\rangle$ is determined by its singularities as $\tau$ approaches rational points on the real line, $\tau \to -\frac{d}{c}$, and these singularities are fixed by modular transformations. We will first use this intuition to guess the answer, and then sketch the rigorous derivation below.

For $c \neq 0$, the $SL(2,\mathbb{Z})$ transformation $\gamma  = \begin{psmallmatrix} f & g \\ c & d \end{psmallmatrix} $ acting on $\tau$ maps the point $ - \frac{d}{c}$ to $i\infty$. For nearby points,
\be
\tau = -\frac{d}{c} + \delta \tau , \qquad |\delta \tau| \ll 1 \ ,
\ee
the transformation rule \eqref{modtheta} gives
\be
\Theta_{H,T}(A^\mu, \tau) \approx |c\tau + d|^{-N} \Theta_{H,T}(\gamma A^\mu, i\infty) \ .
\ee
The corrections to this formula are exponentially suppressed as $\delta \tau \to 0$. Furthermore, the theta function evaluated at $\tau = i\infty$ projects the sum onto contributions with zero scaling dimension (where the scaling dimension is the quantity multiplying $-2\pi \tau_2$ in \eqref{fHS}). In an untwisted theta function, this would project onto the zero vector $n^\mu=0$, corresponding to the ground state of the CFT. In a twisted theta function, $\Theta_{H,T}(a^\mu, b^\mu, i \infty)$ is nonzero only if $b^\mu \in \mathbb{Z}^{2N}$, in which case it is a pure phase $e^{2\pi i \langle a,b \rangle}$. Thus for $\tau \approx - d/c$, the result up to exponentially small corrections is
\begin{align}\label{thetasing}
\Theta_{H,T}(A^\mu,\tau) &\approx |c\tau+d|^{-N} e^{2\pi i 
\langle fa+gb, c a + d b\rangle } \delta(c a^\mu + d b^\mu \in \mathbb{Z})
\end{align}
We use the shorthand
\be
\delta(x \in Y) = \begin{cases} 1 & \mbox{if\ } x \in Y \\ 0 & \mbox{otherwise}\end{cases} \ .
\ee
The right-hand side of \eqref{thetasing} appears to depend on $f,g$, but this must be illusory because the left-hand side depends only on $d/c$. To make this manifest, assume w.l.o.g. that $c$ and $d$ are coprime, and write
\be
f = c m + d^*
\ee
where $m \in \mathbb{Z}$ and $d^*$ is the inverse of $d$ modulo $c$. This is possible because $\det \gamma = f d - g c = 1$. The phase becomes
{\small
\begin{align}
\exp\left[-2\pi i \langle a,b\rangle + \frac{2\pi i}{c} d^* \langle ca+db,ca+db\rangle - \frac{2\pi i}{c} d \langle b,b\rangle+ 2 \pi i m \langle ca+db,ca+db\rangle \right] \ .
\end{align}
}
When the $\delta$-function is satisfied, the quantity multiplying $m$ is $2\pi i \times (\mbox{integer})$ so it drops out, and we now have an expression that depends on only $c$ and $d$.

We are now ready to compute the average. The magic of the Siegel-Weil formula is that the average is simply equal to the sum over singularities, so the natural guess from \eqref{thetasing} is%
\begin{align}\label{generalAverage}
\big\langle \Theta_T(a^\mu, b^\mu, \tau)\big\rangle  &=  \delta(b^\mu \in \mathbb{Z}) + \sum_{(c,d)=1,c>0} |c\tau + d|^{-N} 
 \delta(c a^\mu + d b^\mu \in \mathbb{Z})   \\* 
& \qquad \qquad
\times \exp\left[-2\pi i \langle a,b\rangle + \frac{2\pi i}{c} d^* \langle ca+db,ca+db\rangle - \frac{2\pi i}{c} d \langle b,b\rangle \right]\notag
\end{align}
This turns out to be correct. It is a generalization of the non-holomorphic Eisenstein series. We have included the cusp at infinity separately in the first term.

The answer simplifies for twists living in just one of the two blocks of the Narain inner product in \eqref{narainS}. Suppose that in this basis $a^\mu = (0, a^i)$ and $b^\mu =(0, b^i)$ with $a^i, b^i \in \frac{1}{k}\mathbb{Z}^N$.  This subspace is preserved under modular transformations, and within it all twist inner products vanish, $\langle a,a\rangle = \langle b,b\rangle = \langle a,b\rangle = 0$. Therefore in this case
\begin{align}\label{specialAverage}
\big\langle \Theta_T(a^i, b^i, \tau)\big\rangle  &= \frac{1}{2} \sum_{(c,d)=1} |c\tau + d|^{-N} 
 \delta(c a^i + d b^i \in \mathbb{Z}) \ .
\end{align}
The twists that appear in the deformed WZW models are of this type. In the notation of section \ref{s:wzw},
\begin{align}\label{specialAverage2}
\big\langle \Theta_\Lambda(a, b, \tau)\big\rangle  &= \frac{1}{2} \sum_{(c,d)=1} |c\tau + d|^{-N} 
 \delta(c a + d b \in \Lambda) \ ,
\end{align}
when $a^2=b^2=a\cdot b=0$. In this formula, $\Lambda$ on the right-hand side should be interpreted as the original, undeformed lattice with respect to which $(a,b)$ are defined.
Equivalently,
\begin{align}\label{specialAverage3}
\big\langle \Theta_\Lambda(a, b, \tau)\big\rangle  &= 
(\mbox{Im\ }\tau)^{-N/2} \sum_{\gamma \in \Gamma_{\infty}\backslash \Gamma} (\mbox{Im\ }\gamma \tau)^{N/2}
 \delta(ca + db \in \Lambda) \ ,
\end{align}
where $\Gamma = PSL(2,\mathbb{Z})$ and $\Gamma_{\infty}$ is the abelian group generated by $\begin{psmallmatrix} 1 & 1 \\ 0 & 1 \end{psmallmatrix}$. For prime $k$, this expression simplifies further and becomes an Eisenstein series for the congruence subgroup $\Gamma_0(k)$ as we will discuss in section \ref{ss:spectrum}.

\subsection{Derivation}

We will now sketch the proof of \eqref{generalAverage}. 

We are calculating the average for a self-dual lattice, with both twists $a$ and $b$ non-zero. Siegel solved a slightly different problem: the average for an arbitrary lattice, with only $b$ non-zero. We will show that our case can be reduced to Siegel's. 

First let us state Siegel's main theorem in a convenient way (see \cite{Siegel1,Siegel2,Siegel3} for the original papers, and \cite{siegelbook} for a review in English). Let $S_{\mu\nu}$ be a semi-integral quadratic form of signature $(N,N)$ and choose $b^\mu$ such that $2S_{\mu\nu}b^\nu \in \mathbb{Z}$. The leading behavior of the theta function near rational points $\tau \to - \frac{d}{c}$ (with $c \neq 0$) is of the form
\be
\Theta_{H,S}(0, b^\mu, \tau) \sim |c\tau+d|^{-N} R_S(c,d,b^\mu) \ .
\ee
Crucially, the right-hand side is independent of the moduli matrix $H_{\mu\nu}$. The coefficient $R_S$ is derived by Poisson summation of \eqref{fHS} and is known in closed form, but for our purposes it is enough to know that it exists. (See appendix \ref{app:siegel} for the formula.) Siegel's theorem states that the average over moduli space is the non-holomorphic Eisenstein series defined by summing over singularities,
\be\label{siegelth}
\big\langle \Theta_S(0, b^\mu, \tau) \big\rangle= \delta(b^\mu \in \mathbb{Z}) + \sum_{(c,d)=1, c>0} |c\tau+d|^{-N} R_S(c,d,b^\mu) \ .
\ee
As advertised, the theorem requires $a=0$ but allows for a general quadratic form $S$.

Fortunately we can recast our problem into this form. We want to calculate the average of 
$\Theta_{H,T}(a^\mu, b^\mu, \tau)$
where $ a^\mu \in \frac{1}{k}\mathbb{Z}^{2N}, b^\mu \in \frac{1}{k} \mathbb{Z}^{2N}$ for $k \in \mathbb{Z}_+$, and $T_{\mu\nu}$ is the Narain quadratic form \eqref{narainS}. By decomposing the theta sum into residue classes modulo $k$, $n^\mu = k m^\mu + r^\mu$, we find 
\begin{align}\label{thetaresidues}
\Theta_{H,T}(a^\mu, b^\mu, \tau) 
&= \sum_{r^\mu \in (\mathbb{Z}/k\mathbb{Z})^{2N}} e^{-2\pi i  \langle a, 2r + b\rangle}\Theta_{k^2H,k^2T}( 0, \frac{b^\mu + r^\mu}{k}, \tau) \ .
\end{align}
The twists in the expression on the right-hand side satisfy the assumptions of Siegel's theorem: $a'^\mu = 0$, and $2S_{\mu\nu}'b'^\nu \in \mathbb{Z}$ where $S'_{\mu\nu} = k^2 T_{\mu\nu}$ and $b'^\mu = \frac{1}{k}(b^\mu + r^\mu)$. We can therefore apply \eqref{siegelth} to find the average
{\small
\begin{align}
\big\langle \Theta_T(a^\mu, b^\mu, \tau) \big\rangle &= \delta(b^\mu \in \mathbb{Z}) + \sum_{(c,d)=1,c>0}| c\tau + d|^{-N}  \!\!\!
\sum_{r^\mu \in (\mathbb{Z}/k\mathbb{Z})^{2N}} e^{-2\pi i  \langle  a,2r + b\rangle}
R_{k^2T}(c,d,\frac{b^\mu + r^\mu}{k}) \ .
\end{align}
}
Next, we could plug in Siegel's formula for $R_S$ on the right-hand side, but there is a shortcut. This is a sum over singularities at rational $\tau$. We argued around \eqref{thetasing} that these singularities are fixed by modular invariance of $\Theta_{H,T}$ and are independent of moduli. So there is no need to recalculate them: the final answer must be \eqref{generalAverage}. This completes the derivation.

As a consistency check, we perform the explicit sum of Siegel's $R_{S}$ in appendix \ref{app:siegel} for the case relevant to the WZW model, and (after quite a bit of effort) recover the simple formula \eqref{specialAverage}. In the appendix we also sketch a self-contained derivation of our averaging formula using the method of Maloney and Witten \cite{2006.04839} which turns the averaging integral into a differential equation.

\section{Averaging the deformed WZW model}\label{s:averagewzw}

We now return to the deformed WZW model by assembling the results of sections \ref{s:wzw} and \ref{s:theta}.

\subsection{Partition function}

In section \ref{s:theta} we derived
\be
\big\langle \Theta_{\Lambda}(w(\alpha), w(\beta), \tau) \big\rangle = (\tau_2)^{-N/2} \sum_{\gamma \in \Gamma_\infty \backslash \Gamma} (\mbox{Im\ }\gamma \tau)^{N/2} \delta(c\alpha  + d\beta \in k Q) \ .
\ee
Here we have applied \eqref{specialAverage3} to the WZW model, i.e. to the Narain lattice \eqref{latticewzw} with twists defined in \eqref{twistRel}. To translate the $\delta$-function we used \eqref{abintegers}. Combining this with \eqref{avZsetup} we have
\begin{align}\label{finalZstart}
\big\langle Z_{\rm dWZW}(\tau) \big\rangle = \frac{1}{k^N \tau_2^{N/2} |\eta(\tau)|^{2N}}   \sum_{\alpha,\beta \in Q/kQ} Z_{\alpha,\beta}^{PF}(\tau) 
\sum_{\gamma \in \Gamma_{\infty} \backslash \Gamma} (\mbox{Im\ }\gamma \tau)^{N/2} \delta(c\alpha + d \beta \in kQ) \ .
\end{align}
This is not quite the final answer. The average can only be ascribed a natural holographic interpretation if it can be written in the form $\sum_{\gamma} Z_0(\gamma \tau)$ for some function $Z_0(\tau)$. We will see that it can.

The parafermions satisfy the same modular transformation rule as the bosons,
\be
Z_{\alpha,\beta}^{PF}(\tau) = Z_{\nu, \phi}^{PF}(\gamma \tau)   \quad \mbox{with} \quad \begin{pmatrix}\nu \\ \phi \end{pmatrix} = \gamma \begin{pmatrix} \alpha \\ \beta \end{pmatrix} \ .
\ee
Therefore
\begin{align}
\big\langle Z_{\rm dWZW}(\tau) \big\rangle
&= \frac{1}{k^N \tau_2^{N/2} |\eta(\tau)|^{2N}} \sum_{\alpha,\beta \in Q/kQ} \sum_{\gamma \in \Gamma_\infty \backslash \Gamma} Z^{PF}_{\nu, \phi}(\gamma \tau) (\mbox{Im\ }\gamma \tau)^{N/2} \delta(c\alpha + d \beta \in k Q)
\end{align}
By a simple change of variables we can replace the sum over $\alpha,\beta$ by a sum over $\nu, \phi$,
\begin{align}
\big\langle Z_{\rm dWZW}(\tau) \big\rangle
&= \frac{1}{k^N \tau_2^{N/2} |\eta(\tau)|^{2N}} \sum_{\nu,\phi \in Q/kQ} \sum_{\gamma \in \Gamma_\infty \backslash \Gamma} Z^{PF}_{\nu, \phi}(\gamma \tau) (\mbox{Im\ }\gamma \tau)^{N/2} \delta(\phi \in k Q) \ ,
\end{align}
in which only $\phi = 0$ contributes. Now we plug in the parafermion partition function \eqref{twistedZPF},
\begin{align}
\big\langle Z_{\rm dWZW}(\tau) \big\rangle
&= \frac{1}{(N+1)k^N } \sum_{\lambda \in P_k^+} \sum_{\mu \in P/kQ} \sum_{\nu \in Q/kQ} \sum_{\gamma \in \Gamma_\infty \backslash \Gamma}
 e^{2 \pi i \nu \cdot \mu/k}|c_{\mu}^{\lambda}(\gamma \tau) |^2 
\end{align}
Summing over $\nu$ sets $\mu = 0$ and contributes a factor of $|P/kQ| = (N+1)k^N$, giving our final answer
\begin{align}\label{finalZend}
\big\langle Z_{\rm dWZW}(\tau) \big\rangle &=   \sum_{\gamma \in \Gamma_\infty \backslash \Gamma} \sum_{\lambda \in P_k^+}  |c^\lambda_0(\gamma \tau)|^2 \ .
\end{align}
This is the result quoted in the introduction. 

\subsection{Average of an individual affine representation}
The average contribution of a single WZW primary, $\langle | \chi_{\lambda}(\tau) |^2 \rangle$, also has a Siegel-Weil formula. The string functions for integrable representations at fixed level transform in a finite-dimensional representation $M$ of $SL(2,\mathbb{Z})$ \cite[Chapter 13]{Kac:1990gs}:
\begin{align}
\eta(\tau)^N c^\lambda_\mu(\tau) = \eta(\gamma\tau)^N \sum_{\stackrel{\lambda' \in P_k^+}{\mu' \in P/kQ}}
M^{\lambda \mu'}_{\mu \lambda'}(\gamma) c^{\lambda'}_{\mu'}(\gamma \tau) \ .
\end{align}
Repeating the steps from \eqref{finalZstart} to \eqref{finalZend} we find
\be
\big\langle |\chi_{\lambda}(\tau)|^2 \big\rangle = 
\sum_{\gamma \in \Gamma_\infty \backslash \Gamma} 
\sum_{\stackrel{\lambda_1,\lambda_2 \in P_k^+}{\mu_1,\mu_2 \in P/kQ} }M^{\lambda\mu_1}_{0 \lambda_1}(\gamma) \bar{M}^{\lambda \mu_2}_{0 \lambda_2} (\gamma) c^{\lambda_1}_{\mu_1}(\gamma \tau) \bar{c}^{\lambda_2}_{\mu_2}(\gamma \btau) \ .
\ee
This is a Siegel-Weil formula for the mod-squared characters of an affine Lie algebra.

\subsection{Average spectrum for prime $k$}\label{ss:spectrum}

We will now rewrite the result in a way that makes the spectrum of the average theory easier to discern. The average spectrum will of course be unitary, because we are averaging an ensemble of unitary CFTs with a positive measure. This was not the case in pure gravity \cite{0712.0155}, where the starting point was a Poincar\'e series rather than an ensemble of CFTs, and indeed the Poincar\'e series for pure gravity led to a negative density of states in some regimes \cite{1906.04184}.

We restrict to prime $k$ in this section, for technical reasons that will become clear. As the string functions are complicated functions that in general don't have closed analytical forms, we will do the analysis starting from $\eqref{avZsetup}$, instead of $\eqref{finalZend}$.

The Poincar\'e sum for the averaged theta function \eqref{specialAverage3} only has contributions from coprime $(c,d)$ such that $c \alpha + d \beta \in k Q$. For prime $k$, this condition can only be satisfied when the twists are colinear vectors. Therefore the important averages are
\begin{align}\label{f0fg}
f_0(\tau) &= \big\langle \Theta_{\Lambda}(0,0,\tau) \big\rangle\\
f(\tau) &= \big\langle \Theta_{\Lambda}(w(\alpha), 0, \tau) \big\rangle \notag \qquad (\alpha \neq 0)\\
g(\tau) &= \big\langle \Theta_{\Lambda}(0, w(\beta), \tau) \big\rangle \qquad (\beta \neq 0) \notag
\end{align}
$f$ and $g$ are independent of the twists, because for example the condition $c\alpha \in kQ$ is equivalent to $k | c$ for any $\alpha \neq 0$. The only other non-zero averages for prime $k$ are
\be
g(\tau - h) = \big\langle \Theta_{\Lambda}(w(h\beta), w(\beta), \tau) \ ,
\ee
for $h=0\dots k-1$. 

A useful observation is that the basic ingredients $f_0, f, g$ are all proportional to Eisenstein series. The untwisted term $f_0$ is the non-holomorphic Eisenstein series for $SL(2,\mathbb{Z})$; this is the same quantity that appears in the Narain duality \cite{2006.04839,2006.04855}.
The twisted average $f$ is proportional to an Eisenstein series for the congruence subgroup \begin{equation}
    \Gamma_0(k)=\left\{\left.\left(\begin{array}{cc}
        a & b \\
        c & d
    \end{array}\right)\in \Gamma\right|c=0\textrm{ mod }k\right\} \ .
\end{equation}
These Eisenstein series depend on a choice of Dirichlet character mod $k$. 
For the character $\chi$, the real-analytic Eisenstein series is defined as
\begin{align}
E_k(\tau, s; \chi) &= \sum_{\gamma \in \Gamma_{\infty} \backslash \Gamma_0(k)} \overline{\chi(d)} (\mbox{Im\ }\gamma \tau)^s \ .
\end{align}
Elements of $\Gamma_{\infty} \backslash \Gamma_0(k)$ correspond to coprime pairs $(c,d)$ with $c = 0 \mod k$. Therefore, for the principal character $\chi_0(d)$, defined to be 1 if $(d,k) = 1$ and 0 otherwise, we have
\begin{align}
 E_k(\tau, s) := E_k(\tau, s; \chi_0) &= \sum_{\gamma \in \Gamma_{\infty} \backslash \Gamma_0(k)} (\mbox{Im\ }\gamma \tau)^s
 \\
&= \frac{1}{2}(\mbox{Im\ }\tau)^s \sum_{(c,d)=1} |c k\tau + d|^{-2s} \chi_0(d) \ .
\end{align}
To relate this to the averaged theta functions, first consider $\alpha \neq 0$, $\beta = 0$. We found the average
\begin{equation}
    \begin{aligned}
f(\tau) =  \big\langle \Theta_{\Lambda}(w(\alpha), 0, \tau) \big\rangle &=1+\sum_{(c,d)=1,c>0} |c \tau+ d|^{-N} \delta(c \alpha \in k Q)\\
&= (\mbox{Im\ }\tau)^{-N/2} E_k(\tau, \frac{N}{2}) \ .
\end{aligned}
\end{equation}
Similarly for nonzero $\beta$,
\begin{align}
g(\tau) =  \big\langle \Theta_{\Lambda}(0, w(\beta), \tau) \big\rangle 
&= \sum_{(c,d)=1,c>0} |c \tau + k d|^{-N} \chi_0(c) \ ,
\end{align}
which is proportional to  $E_k(-\frac{1}{\tau}, \frac{N}{2})$.
The Eisenstein series at prime level can be easily expressed in terms of the Eisenstein series for $SL(2,\mathbb{Z})$. In appendix \ref{app:fourier} we derive the relations
\begin{align}\label{ff0relmain}
f(\tau) &= \frac{1}{k^N-1}\left(k^N f_0(k\tau) - f_0(\tau)\right) \\
g(\tau) &= \frac{1}{k^N-1}\left( f_0( \frac{\tau}{k} ) - f_0(\tau) \right) \ .\notag
\end{align}
Now we will use these results to rewrite $\langle Z_{\rm dWZW} \rangle$ in a way that makes the spectrum more apparent. Recall that the averaged partition function is
\begin{align}\label{bvZsetup}
\langle Z_{\rm dWZW} (\tau) \rangle &= k^{-N}  |\eta(\tau)|^{-2N} \sum_{\alpha,\beta \in Q/kQ} Z^{PF}_{\alpha,\beta}(\tau) \big\langle \Theta_{\Lambda}(w(\alpha), w(\beta), \tau) \big\rangle   \ .
\end{align}
At prime $k$,
{\small
\begin{align}
\big\langle Z_{\rm dWZW}(\tau) \big\rangle &= \frac{1}{k^N |\eta(\tau)|^{2N}} \left( 
Z_{0,0}^{PF}(\tau) f_0(\tau) + f(\tau) \sum_{\stackrel{\alpha \in Q/kQ}{\alpha \neq 0}} Z_{\alpha,0}^{PF}(\tau)  + \sum_{\stackrel{\alpha \in Q/kQ}{\alpha \neq 0}} \sum_{h=0}^{k-1} Z_{h\alpha, \alpha}^{PF}(\tau) g(\tau-h) \right) \ .
\end{align}
}
Plugging in the twisted parafermion partition function \eqref{twistedZPF}, this becomes
\begin{align}\label{zccnice}
\big\langle Z_{\rm dWZW}(\tau) \big\rangle &=
\sum_{\lambda \in P_k^+}\left(
h_1(\tau) |c^{\lambda}_0(\tau)|^2
 + h_2(\tau) \sum_{\mu \neq 0} |c^\lambda_\mu(\tau)|^2 
  + \sum_{\alpha\neq 0, \mu} c^{\lambda}_\mu(\tau) \bar{c}^{\lambda}_{\mu - \alpha} (\btau)
  h_3^{\mu, \alpha}(\tau)
\right)
\end{align}
where
\begin{align}
h_1(\tau) &= \frac{1}{k^N}\left( f_0(\tau) +(k^N-1) f(\tau) \right) = f_0(k\tau)  \label{h1}\\
h_2(\tau) &= \frac{1}{k^N}\left( f_0(\tau) - f(\tau) \right)  = \frac{1}{k^N-1}( f_0(\tau) - f_0(k\tau)) \label{h2}\\
h_3^{\mu, \alpha}(\tau) &=   \frac{1}{k^N} \sum_{h=0}^{k-1} e^{ i \pi h \alpha \cdot(2\mu - \alpha)/k} g(\tau-h)
\end{align}
The full spectrum is thus a convolution of the parafermion spectrum with that of $h_1$, $h_2$, and $h_3^{\mu,\alpha}$. The latter spectra can be found by a Fourier transform. 
The Fourier expansions of $h_1$ and $h_2$ are given in \eqref{goodh1} and \eqref{goodh2}. They are manifestly positive, and they have continuous spectra starting at the unitarity bound. 
Finally, consider the $h_3$ term in \eqref{zccnice}. The Fourier transform of $g(\tau)$ is 
\be
g(\tau) = \sum_{l \in \mathbb{Z}} g_\ell(\tau_2) e^{2\pi i l \tau_1/k}
\ee
where the coefficients $g_\ell$ are given in \eqref{goodg0} and \eqref{goodgl} and are also manifestly positive.
Performing the sum over $h$, we get
\be
\sum_{\lambda \in P_k^+} \sum_{\alpha\neq 0, \mu} c^{\lambda}_\mu \bar{c}^{\lambda}_{\mu - \alpha} 
  h_3^{\mu, \alpha}(\tau)=k^{1-N} \sum_{\lambda,\mu,\alpha\neq 0} \sum_{l=-\infty}^{\infty} c_{\mu}^{\lambda} \overline{c}_{\mu-\alpha}^{\lambda}  g_l(\tau_2) e^{2\pi i l \frac{\tau_1}{k}} \delta(l-\alpha\cdot\mu+\alpha^2/2  \in k \mathbb{Z})
\ee
Although the $\exp(2\pi i l\tau_1/k)$ is not periodic under $\tau_1\to \tau_1+1$, the full expression is periodic indicating a spectrum with integers spins, because the string functions obey (see for example \cite[\S 14.5]{DiFrancesco:1997nk})
\be
c_{\mu}^\lambda(\tau+1) \bar{c}_{\mu-\alpha}^{\lambda}(\btau+1) = e^{ i\pi (\alpha^2 - 2 \mu \cdot \alpha)/k} c_{\mu}^\lambda(\tau) \bar{c}_{\mu-\alpha}^{\lambda}(\btau) \ .
\ee

To summarize, for prime $k$, the average partition function is \eqref{zccnice}, where the coefficients $h_1, h_2, h_3^{\mu,\alpha}$ are all known explicitly and we have given their Fourier expansions. The spectrum is unitary, continuous, and consists only of integers spins. Of course the the spectrum of the deformed WZW model is unitary and integer-spin at any point in moduli space, so these features were guaranteed by the setup. The continuous spectrum is a consequence of averaging.

\section{Discussion of the holographic interpretation}\label{s:discussion}

A natural ansatz for the partition function of 3d gravity is  the Poincar\'{e} series \cite{0712.0155}
\be\label{zpoincare}
Z(\tau) = \sum_{\gamma \in \Gamma_\infty \backslash \Gamma} Z_{0}(\gamma \tau) \ ,
\ee
where $Z_0(\tau)$ is the perturbative partition function of bulk fields on the solid torus with complex structure $\tau$ at infinity. The solid torus is a Euclidean black hole, and the Poincar\'{e} sum is interpreted as a sum over the Maldacena-Strominger family of $SL(2,\mathbb{Z})$ black holes \cite{hep-th/9804085}. In a general theory of 3d gravity there are additional contributions to the path integral, but \eqref{zpoincare} might be sufficient (or at least dominant) in simple examples. The averaged partition function of Narain CFTs takes this form; this observation is the starting point for the Narain duality proposed in \cite{2006.04839, 2006.04855} (see also \cite{1111.1987, 1907.06656} for earlier efforts to apply Poincar\'{e} series to rational CFTs).
A similar sum can be used to calculate indices in supersymmetric theories \cite{hep-th/0005003, 1111.1161}, but in the present discussion, $Z$ is not holomorphic in $\tau$. 

We have shown that the averaged partition function in the deformed WZW model $\langle Z_{\rm dWZW}\rangle$ can also be written as a Poincar\'{e} series. This suggests a holographic interpretation. To succeed, we need a three-dimensional theory for which the perturbative partition function on a solid torus is
\be
Z_{0}(\tau) = \sum_{\lambda \in P_k^+} |c^\lambda_0(\tau)|^2 \ .
\ee
Let us reorganize this into
\be\label{zwithb}
Z_{0}(\tau) =  \frac{1}{|\eta(\tau)|^{2N} } \sum_{\lambda \in P_k^+} |b^\lambda_0(\tau)|^2 \ ,
\ee
where $b^\lambda_\mu(\tau) = \eta(\tau)^N c^\lambda_\mu(\tau)$ is the branching function of the $SU(N+1)_k/U(1)^N$ coset, i.e., the parafermion character. We have factored out the term $|\eta(\tau)|^{-2N}$, which is the perturbative partition function of three-dimensional $U(1)^{2N}$ Chern-Simons theory \cite{1903.05100,2006.04839, 2006.04855}. For $k=1$, the branching functions are trivial and there is only one non-zero contribution to the sum over $\lambda$, so we recover the Narain duality. This is expected, because $SU(N+1)_1$ lives in the moduli space of $N$ free bosons. 

The form \eqref{zwithb} suggests that at the perturbative level, the holographic dual of the average deformed WZW model is $U(1)^{2N}$ Chern-Simons theory coupled to topological `matter fields'. The matter fields should account for the parafermion contribution to \eqref{zwithb}. There is a well known construction of holomorphic parafermion characters using Chern-Simons theory \cite{Witten:1988hf,Witten:1991mm,Isidro:1991fp,Gawedzki:2001ye}, but what is needed is a non-holomorphic combination of left and right movers. For abelian Chern-Simons theory this was achieved in \cite{hep-th/0403225}, but to our knowledge the nonabelian generalization needed to write a Lagrangian for our bulk theory does not appear in the literature, so we defer this problem to the future.\footnote{An alternative approach would be to implement the orbifold $\Gamma_k$ directly on the bulk theory, as in \cite{2103.15826}.}

In addition to finding the bulk Lagrangian, there are several concrete directions to pursue. The duality can be explored at higher genus, as in \cite{2006.04855}; the double torus is especially interesting because it encodes energy spacing statistics of the ensemble \cite{2006.08648}. As mentioned in the introduction, there are nontrivial correlation functions in the deformed WZW duality that can be studied using methods similar to \cite{2103.15826}. It may be interesting to analyze the thermodynamics of the averaged theory as well. In particular, we have not calculated the spectral gap or compared to the Cardy formula (see \cite{2006.04839} for the Narain case). Part of our motivation for studying this class of theories is that the symmetry class of theories with $U(1)^N$ chiral algebra and $c>N$ interpolates between the pure Virasoro case and the maximally symmetric case $c=N$ relevant to sphere packing \cite{1905.01319,2006.02560}; it would be interesting to understand whether this ultimately gives insight into the pure Virasoro limit, relevant to Einstein gravity in three dimensions.

We conclude with a few more general open questions:

\begin{enumerate}
\item Is there a non-perturbative definition of the bulk theory, in the Narain duality or the duality described here, that does not rely on an ad hoc prescription for the sum over manifolds?
\item Is there a worldsheet interpretation of the ensemble average? The WZW models in particular appear as worldsheet sigma models. It would also be very interesting to average over Calabi-Yau moduli, but this requires new ideas \cite{2103.15826}.
\item Can these methods be extended to non-compact gauge groups? If so, they can be applied to the $SL(2,R)/U(1)$ black hole \cite{Witten:1991yr}, or string theory on AdS$_3$ \cite{hep-th/0001053}. Current-current deformations of non-compact CFTs have been studied in, e.g., \cite{Israel:2003ry, Orlando:2005vt, Detournay:2010rh, Giribet:2021cyy}.

\item Can these holographic dualities be embedded into string theory? If so, what is the meaning of the ensemble average? In string theory we expect the dual to be a single CFT, but it could be useful to introduce an explicit average to calculate certain observables. For example, maybe the average allows some statistical properties of the UV theory to be accessed in supergravity. 
\end{enumerate}

\ \\
\bigskip

\noindent \textbf{Acknowledgments} We thank Nathan Benjamin, Ryan Bilotta, Henry Cohn, Minli Feng, Chao-Ming Jian, Alex Maloney, Hirosi Ooguri, Edgar Shaghoulian, and Edward Witten for  discussions. This work is supported by the Simons Foundation through the Simons Collaboration on the Non-perturbative Bootstrap.

\appendix

\section{Details on orbifolding}\label{app:orbifold}

\subsection{Comparison to F\"orste and Roggenkamp}
In section \ref{s:wzw}, we constructed the WZW model as an orbifold by the group $\Gamma_{k} \cong Q/kQ$. This is based on \cite{hep-th/0304234}, but our orbifold action appears to be slightly different, so here we will explain the differences. 

First, the orbifold group in \cite{hep-th/0304234} is described as $P/kQ$. This is not a significant difference, because the elements of $P/Q$ correspond to the center of $SU(N+1)$ and act trivially. In particular, $Z^{PF}_{\alpha,\beta}$ vanishes for either $\alpha$ or $\beta$ in $P/Q$, so \eqref{zOrb} can be written as a sum over $\alpha, \beta \in P/kQ$, in agreement with \cite{hep-th/0304234}.

More importantly, the orbifold action in \cite{hep-th/0304234} (see also earlier work by Gepner \cite{Gepner:1987sm}) is defined to be \eqref{gammaaction} without the $B^0$ terms. This leads to ambiguities, because the resulting phase depends on the coset representatives chosen for $\gamma, \mu, \bar{\mu}$.  Our definition of the symmetries resolves these ambiguities. To see this, rewrite the phase in \eqref{gammaaction} as
$\exp\left( \frac{2\pi i }{k} (\gamma^i R^+_{ij} \mu^j + \gamma^i R^-_{ij} \bar{\mu}^j)  \right)$ where $R^{\pm}_{ij} = \frac{1}{2}(G^0_{ij} \mp B^0_{ij})$. The matrices $R^{\pm}_{ij}$ have integer entries, so this is invariant under shifting $\gamma, \mu$, or $\bar{\mu}$ by an element of $kQ$. Without $B^0$, there would be a sign ambiguity under these shifts.\footnote{The twisted partition functions can be calculated from the orbifold symmetries by a standard procedure: (i) Insert the symmetry into the trace $Z_{\alpha, 0} = \mbox{Tr\ }g_{\alpha} q^{L_0-c/24} \bar{q}^{\bar{L}_0-c/24}$, (ii) Act with $\tau \to -1/\tau$ to find $Z_{0, \alpha}$, (iii) Insert a different symmetry to find $Z_{\beta,\alpha}$. We came upon \eqref{gammaaction} as follows. First we calculated the partition function of $N$ free bosons with twisted boundary conditions, using standard path integral methods. This leads to \eqref{twistedZB}, and from here one can easily work backwards through this 3-step procedure to determine the orbifold symmetries. }


In the rest of this appendix, we elaborate on the orbifold sum, Eq. (\ref{zOrb}).\par
\subsection{Vacuum Normalization}
First, let us explain the normalization of the parafermion partition function in Eq. (\ref{twistedZPF}). The coefficient, $(N+1)^{-1}$, is chosen to normalize the vacuum state to have unit coefficient (see \cite[\S 18.2]{DiFrancesco:1997nk}). Consider the untwisted partition function of the parafermions:
\begin{equation}
Z^{PF}_{0,0}(\tau) = 
\frac{1}{N+1} |\eta(\tau)|^{2N} \sum_{\lambda \in P_k^+} \sum_{\mu \in P/kQ} |c_\mu^\lambda(\tau)|^2
\end{equation}
Here we prove a useful identity: 
\begin{equation}
\sum_{\lambda \in P_k^+}\sum_{\mu \in P/kQ} c^{\lambda}_{\mu}\bar{c}^{\lambda}_{\mu-\beta}=(N+1)\sum_{\lambda \in P_k^+}\sum_{\mu \in P/kP} c^{\lambda}_{\mu}\bar{c}^{\lambda}_{\mu-\beta} \ .
\end{equation}
This identity can be shown by applying outer automorphisms. The automorphism group is $O(\hat{su}(N+1))=\mathbb{Z}_{N+1}$ whose generator we denote by $a$. The group acts on affine weights, but we only care about its action on the finite part of the affine weights. For $A=a^n\in O(\hat{su}(N+1))$ with $n=0\dots N$, this is
\begin{equation}
    a^n\lambda=ke^n+w_{a^n}\lambda
\end{equation}
in which $\{e^i\}$ are the fundamental weights, $e^0 = 0$, and $w_{a^n}$ is a Weyl transformation corresponding to $a^n$. 
We note that $\{0\}\cup\{ke^i\}=kP/kQ$. Thus, for a given $n$ and $\mu$, we have 
\begin{align}
    \sum_{\lambda\in P_k^+}c^\lambda_{\mu+ke^n}\bar{c}^\lambda_{\mu+ke^n-\beta}=&\sum_{\lambda\in P_k^+} c^{a^{-n}\lambda}_{w_{a^n}^{-1}\mu}\bar{c}^{a^{-n}\lambda}_{w_{a^n}^{-1}(\mu-\beta)}\\
    =&\sum_{\lambda\in P_k^+} c^{a^{-n}\lambda}_{\mu}\bar{c}^{a^{-n}\lambda}_{\mu-\beta}\notag\\
    =&\sum_{\lambda\in P_k^+}c^{\lambda}_{\mu}\bar{c}^{\lambda}_{\mu-\beta}\notag
\end{align}
in which we used the fact that $c^{A\lambda}_{A\mu}=c^\lambda_\mu$ in the second equality and $c^\lambda_{w\mu}=c^\lambda_\mu$ in the third equality. The last equality simply follows from $a^n$ being a symmetry of the set $P_k^+$. The equality above means that, for a particular sum of string functions of this kind, we are able to shift its $\mu$ index by $ke^i$ without changing its value. Thus,

\begin{align}
(N+1)\sum_{\lambda \in P_k^+}\sum_{\mu \in P/kP} c^{\lambda}_{\mu}\bar{c}^{\lambda}_{\mu-\beta}=\sum_{\lambda \in P_k^+}\sum_{\delta\in kP/kQ}\sum_{\mu \in P/kP} c^{\lambda}_{\mu+\delta}\bar{c}^{\lambda}_{\mu+\delta-\beta}=\sum_{\lambda \in P_k^+}\sum_{\mu \in P/kQ} c^{\lambda}_{\mu}\bar{c}^{\lambda}_{\mu-\beta}
\end{align}
We can thus write
\begin{equation}\label{kPsum}
Z^{PF}_{0,0}(\tau) = 
|\eta(\tau)|^{2N} \sum_{\lambda \in P_k^+} \sum_{\mu \in P/kP} |c_\mu^\lambda(\tau)|^2
\end{equation}

Finally, if we limit ourselves to $\mu\in P/kP$, the vacuum term $(q \bar{q})^{-c_{PF}/24}$ comes entirely from the $|c_0^0|$ term, whose coefficient is $1$. Thus, the partition function is correctly normalized.\par
\subsection{Orbifold Sum}
Now we derive Eq. (\ref{zOrb}). Using the identity Eq. (\ref{kPsum}), we are able to write
{\small\begin{align}
    k^{-N} \sum_{\alpha,\beta \in Q/kQ} &Z_{\alpha,\beta}^{PF}(\tau) Z^\Lambda_{\alpha,\beta}(\tau) \notag \\
    &=k^{-N}\sum_{\alpha,\beta\in Q/kQ}\sum_{\lambda\in P_k^+}\sum_{\mu'\in P/kP}\sum_{\mu\in P/kQ}e^{\pi i \alpha\cdot (2\mu-2\mu')}c^{\lambda}_{\mu'}(\tau)\bar{c}^\lambda_{\mu'-\beta}(\bar{\tau})\theta_{\mu}(\tau)\bar{\theta}_{\mu-\beta}(\bar{\tau})\notag\\
    &=\sum_{\beta\in Q/kQ}\sum_{\lambda\in P_k^+}\sum_{\mu'\in P/kP}\sum_{\mu\in P/kQ}\delta(\mu-\mu'\in k P)c^{\lambda}_{\mu'}(\tau)\bar{c}^\lambda_{\mu'-\beta}(\bar{\tau})\theta_{\mu}(\tau)\bar{\theta}_{\mu-\beta}(\bar{\tau})\notag\\
    &=\sum_{\beta\in Q/kQ}\sum_{\lambda\in P_k^+}\sum_{\mu\in P/kQ}c^{\lambda}_{\mu}(\tau)\bar{c}^\lambda_{\mu-\beta}(\bar{\tau})\theta_{\mu}(\tau)\bar{\theta}_{\mu-\beta}(\bar{\tau})\notag\\
    &=\sum_{\beta\in P/kQ}\sum_{\lambda\in P_k^+}\sum_{\mu\in P/kQ}c^{\lambda}_{\mu}(\tau)\bar{c}^\lambda_{\mu-\beta}(\bar{\tau})\theta_{\mu}(\tau)\bar{\theta}_{\mu-\beta}(\bar{\tau})\notag\\
    &=\sum_{\lambda\in P_k^+}\left|\sum_{\mu\in P/kQ}c^{\lambda}_{\mu}(\tau)\theta_\mu(\tau)\right|^2\notag\\
    &=\sum_{\lambda\in P_k^+}|\chi_\lambda(\tau)|^2\notag\\
    &=Z_{\rm WZW}(\tau) \ .
\end{align}
}
In line $5$ we extended the range of the $\beta$ summation because the product of string functions vanish if $\beta\not\in Q$. 

\section{Details on Narain lattices}\label{app:narain}

Denote a vector in $\lambda \in \mathbb{R}^{N,N}$ by $\lambda = (\lambda_+, \lambda_-)$, where $\lambda_+$ and $\lambda_-$ are each vectors in $\mathbb{R}^N$. The inner product is
\be
\lambda \cdot \mu = \lambda_+ \cdot \mu_+ - \lambda_- \cdot \lambda_- \ .
\ee
For vectors in $\mathbb{R}^N$, the inner product is the ordinary Euclidean dot product.

A Narain lattice is an even self-dual lattice. Any Narain lattice in $\mathbb{R}^{N,N}$ can be constructed starting from a Euclidean lattice $L \subset \mathbb{R}^N$ as follows. Start with a basis $v_i$ for $L$, and dual basis $v^i$ for $L^*$ satisfying $v^i \cdot v_j = \delta^{i}_j$. (Each $v_i$ is a vector in $\mathbb{R}^N$ with components $v_i^a$, but we will suppress the lattice index $a$ which is raised and lowered with $\delta_{ab}$.) The basis for $L$ defines a metric
\be
G_{ij} = v_i \cdot v_j \ .
\ee
Define the $\mathbb{R}^{N,N}$ vectors
\begin{align}
V^i &= (v^i, v^i) \\
\hat{V}_i &= \frac{1}{2} (v_i + B_{ij}v^j, -v_i + B_{ij} v^j) \ ,
\end{align}
where $B_{ij} = -B_{ji}$. The set of $2N$ vectors $\{ V^i, \, \hat{V}_j \}$ is a basis for a Narain lattice, $\Omega$. The basis vectors satisfy
\be
V^i \cdot \hat{V}_j = \delta^i_j , \quad V^i \cdot V^j = \hat{V}_i \cdot \hat{V}_j = 0 \ ,
\ee
which makes it easy to check that $\Omega$ is even and self-dual. The moduli of the Narain lattice are the antisymmetric tensor $B_{ij}$ and the metric $G_{ij}$. There are also $O(N) \times O(N)$ rotations, but these do not affect the CFT or the theta function.

Introduce an index $\mu = 1, \dots, 2N$ running over all the Narain basis vectors, 
\be
U_\mu = (V^i, \hat{V}_i) \ .
\ee
Points on the lattice are
\be
\lambda = n^\mu U_\mu = w_iV^i + p^i \hat{V}_i \ ,
\ee
with $n^\mu \in \mathbb{Z}^{2N}$. The integers $p^i$ and $w_i$ are momentum and winding of the compact boson.

The Narain lattice $\Lambda$ associated to the level-$k$ WZW model was defined in \eqref{latticewzw}. The corresponding Euclidean lattice is $L = \sqrt{k} Q$, so in terms of the roots and weights, $v_i = \sqrt{k}e_i$, $v^i = \frac{1}{\sqrt{k}}e^i$. The moduli are   \cite{Ginsparg:1986bx} 
\be
G_{ij} = k G^0_{ij} \ , \qquad
B_{ij} = k B^0_{ij} 
\ee
where $G_{ij}^0$ is the Cartan matrix and $B^0_{ij}$ was defined below \eqref{B0}.
The basis vectors are
\be
V^i = \frac{1}{\sqrt{k}}(e^i, e^i) \ , \quad
\hat{V}_i = \frac{\sqrt{k}}{2}( e_i + B^0_{ij} e^j, -e_i + B^0_{ij} e^j ) \ .
\ee
The twists \eqref{twistRel} are therefore 
\be
w(\alpha) = \frac{1}{k}  \alpha^i \hat{V}_i \ , \quad
w(\beta) = \frac{1}{k} \beta^i \hat{V}_i 
\ee
where $\alpha^i  = e^i \cdot \alpha$, $\beta^i = e^i \cdot \beta$. In the notation used in section \ref{s:theta}, $a = a^\mu U_\mu$ and $b = b^\mu U_\mu$, the components are
\be\label{abintegers}
a^\mu(\alpha) = ( \vec{0}, \frac{1}{k} \alpha^i ) \ , \quad
b^\mu(\beta) = (\vec{0}, \frac{1}{k} \beta^i) \ .
\ee

\section{Details on Siegel's formula}\label{app:siegel}

In this appendix we give some more details on Siegel's averaging formula, and use it to reproduce our average \eqref{specialAverage} directly from Siegel's theorem as stated in \cite{siegelbook}. We also rederive our result using the method of Maloney and Witten based on a modular differential equation.

\subsection{The Averaging of partition functions $\Theta_{H,T}(0,b^\mu,\tau)$}
We start with the case $a = 0$. Note that $\Theta_{H,T}(0,b^\mu,\tau)=\Theta_{kH,kT}(0,b^\mu,\frac{\tau}{k})$. This rescaling is necessary so that $kT b^\mu$ has integer entries to apply Siegel's formula. In our basis, $b^\mu=(0,b^i)$.\par
The Siegel-Weil formula \cite{Siegel1,Siegel2,Siegel3,siegelbook} tells us that
\begin{align}
     &\braket{\Theta_T(0,b^\mu,\tau)}=\int dH \Theta_{kH,kT}(0,b^\mu,\frac{\tau}{k})\\
     &\ \ =\gamma_{b}+\frac{1}{k^N} \sum_{(c,d)=1,d>0} d^{-2N} \left|\frac{\tau}{k}-\frac{c}{d}\right|^{-N} R_{S}(c,d,b^\mu) \notag
\end{align}
in which 
\begin{equation}
R_S(c,d,b^\mu)= \sum_{g \in (\mathbb{Z}/d\mathbb{Z})^{2N}} \exp(2\pi i \frac{c}{d} k\braket{g+b,g+b}) 
\end{equation}
is the singularity coefficient mentioned in Eq. (\ref{siegelth}), and derived in \cite{siegelbook}.
$\gamma_b=0$ unless $b=0$, in which case $\gamma_b=1$. We note that pairs $(c,d)$ where $d>0$ can be mapped to elements in $\Gamma_\infty \backslash SL(2,\mathbb{Z})$. That is, we write matrices in $SL(2,\mathbb{Z})$ as
\begin{equation}
    \gamma=\left(\begin{array}{cc}
        f & g \\
        c & -d
    \end{array}\right) \ .
\end{equation}
We simplify the last sum of phases as follows. First we write $g=(m,n)$ in which $m,n$ are $N$ dimensional integer mod $d$ vectors:
\begin{align}
R_S(c,d,b^\mu)=&\sum_{g \in (\mathbb{Z}/d\mathbb{Z})^{2N}}\exp(2 \pi i \frac{c}{d} k \braket{g+b,g+b}))\\ 
=&\sum_{m,n \in (\mathbb{Z}/d\mathbb{Z})^{N}}\exp(2 \pi i \frac{c}{d}km\cdot(n+b))\notag\\ 
=&\prod_{i=1}^N\sum_{m_i,n^i=1}^{d}\exp(2 \pi i \frac{c}{d}k m_i(n^i+b^i))\notag\\ 
=&\prod_{i=1}^N\sum_{n_i=1}^{d}\frac{1-\exp(2 \pi i c k b^i)}{1-\exp(2 \pi i \frac{c}{d}k(n^i+b^i))} \notag
\end{align}
This expression is zero if the denominator is not zero because the numerator is always zero. When the denominator is zero, the fraction is $d$ by l'Hospital's rule. Thus
\begin{align}
R_S(c,d,b^\mu)=d^N\prod_{i=1}^N \# \{ n^i \in (\mathbb{Z}/d \mathbb{Z})^N \  | \   kn^i+b^i\text{ mod } d=0 \}
\end{align}
Now we have a lemma:
\begin{lemma}
The number of solutions of 
\begin{equation}
kx+b\text{ mod } d=0
\end{equation}
is equal to $(k, d)$ when $(k,d)|b$; otherwise it is zero.
\end{lemma}

\begin{proof}
When $(k,d)$ does not divide $b$, the equation obviously has no solutions; any solution would introduce a contradiction because one side is divisible by $(k,d)$ and the other side is not.\par
Otherwise consider first the case $(k,d)=1$. The proposition is obvious because $k$ is  invertible mod $d$.\par
Suppose $a=(k,d)\neq 1$. Define $k'=k/a,d'=d/a,b'=b/a$. Then the equation
\begin{equation}
k' x+b' \text{ mod } d'=0
\end{equation}
has exactly one solution, $x_0$, with $k' x_0+b'=c d'$, or $k x_0+b=c d$.\par
Consider any solution $y$ of the original equation, $k y+b=d $. This means that $k' (y-x_0)=(1-c) d'$, 
so that $y-x_0$ is divisible by $d'$. Thus, $y=x_0+ p d'$; since $y\in \mathbb{Z}/d \mathbb{Z}$, $p$ can take values from $1$ to $a$. Therefore there are $a=(d,k)$ solutions.\par
\end{proof}
In the end, we are able to collect all the terms: 
{\small
\begin{align}
    \braket{\Theta_T(0,b^\mu,\tau)}&=\gamma_{b}+\frac{1}{k^N} \sum_{(c,d)=1,d>0} d^{-N} \left|\frac{\tau}{k}-\frac{c}{d}\right|^{-N} \prod_{i=1}^N \# \{ n^i \in (\mathbb{Z}/d \mathbb{Z})^N \  | \   kn^i+b^i\text{ mod } d=0 \}\\
    &=\gamma_{b}+\frac{1}{k^N} \sum_{(c,d)=1,d>0} d^{-N} \left|\frac{\tau}{k}-\frac{c}{d}\right|^{-N}(k,d)^{N} \delta(kb^i\text{ mod } (k,d))\notag\\ 
    &=\gamma_{b}+\frac{1}{k^N} \sum_{(c,d)=1,d>0} d^{-N} \left|\frac{\tau}{k}-\frac{c}{d}\right|^{-N} \delta(ckb^i\text{ mod } (k,d))\notag
\end{align}
}
The last line is true because $b$ has entries whose numerators are less than or equal to $k$, and $(c,d)=1$. Now we write $k=(k,d)p,d=(k,d)q,(p,q)=1$. Furthermore we note that $(pc,q)=1$. Then
\begin{align}
    \braket{\Theta_T(0,b^\mu,\tau)}&=\gamma_b+ \sum_{(pc,q)=1,q>0} |q \tau-p c|^{-N}  \delta(pckb^i\text{ mod } k) \\
    &=\gamma_b+ \sum_{(c,d)=1,c>0} |c\tau- d|^{-N}  \delta (d k b^i=0\text{ mod } k)\notag\\
    &=\gamma_b+ \sum_{(c,d)=1,c>0} |c\tau- d|^{-N}  \delta (db^i\in \mathbb{Z}^N)\notag
\end{align}
in which in the second equality we have renamed all the variables. It is worth mentioning that the renaming of the variables do not cause any multiplicity troubles. That is, there is a bijection between pairs $(c,d)$ and pairs $(pc,q)$: the map from $(c,d)$ to $(pc,q)$ is the definition. The map from $(pc,q)$ to $(c,d) $ is by observing that $(k,pc)=p(k/p,c)=p((k,d),c)=p$, so $(k,d)=k/(k,pc)$ and thus $d=k/(k,pc)q$, $c=pc/(k,pc)$.\par
This completes the averaging process for $\Theta_{H,T}(0, b^\mu, \tau)$.\par
\subsection{The Averaging of partition functions $\Theta_{H,T}(a^\mu, b^\mu, \tau)$}
Now consider the case $a \neq 0$. We note the following identity of theta functions:
\begin{equation}
\Theta_{H,T}(a^\mu, b^\mu, \tau) 
= \sum_{r^\mu \in (\mathbb{Z}/k\mathbb{Z})^{2N}} e^{-2\pi i  \langle a, 2r + b\rangle}\Theta_{k^3H,k^3T}( 0, \frac{b^\mu + r^\mu}{k}, \frac{\tau}{k}) 
\end{equation}
Here we take the basis $a^\mu=(0,a^i),b^\mu=(0,b^i)$. We have thus reduced the averaging to the procedure in the previous section. Note that an $k^2$ rescaling as in \eqref{thetaresidues} would have been sufficient for the requirement of the averaging process because $2k^2 T(b+r)/k$ has integer entries; we used a $k^3$ scaling to simplify the following calculations.\par
We can thus calculate the average
{\small
\begin{align}
    \braket{\Theta_{T}(a^\mu, b^\mu, \tau)}=\gamma_b+\frac{1}{k^{3N}}& \sum_{(c,d)=1,d>0} d^{-2N} |\frac{\tau}{k}-\frac{c}{d}|^{-N} \\& \times\sum_{r \in (\mathbb{Z}/k\mathbb{Z})^{2N}}\sum_{g \in (\mathbb{Z}/d\mathbb{Z})^{2N}} \exp\left(2\pi i \frac{c}{d} k^3\braket{g+\frac{r}{k}+\frac{b}{k},g+\frac{r}{k}+\frac{b}{k}}-2\pi i  \langle a, 2r\rangle\right)\notag
\end{align}
}
Again, we focus on the sum of phases.
\begin{align}
    &\sum_{r \in (\mathbb{Z}/k\mathbb{Z})^{2N}}\sum_{g \in (\mathbb{Z}/d\mathbb{Z})^{2N}} \exp(2\pi i \frac{c}{d} k^3\braket{g+\frac{r}{k}+\frac{b}{k},g+\frac{r}{k}+\frac{b}{k}}-2\pi i  \langle a, 2r\rangle)\\
    =&\sum_{r \in (\mathbb{Z}/k\mathbb{Z})^{2N}}\sum_{g \in (\mathbb{Z}/d\mathbb{Z})^{2N}} \exp(2\pi i \frac{c}{d} k\braket{kg+r+b-\frac{da}{ck},kg+r+b-\frac{da}{ck}})\notag\\ 
    =&\sum_{m\in (\mathbb{Z}/dk\mathbb{Z})^{2N}}\exp(2\pi i \frac{c}{d} k\braket{m+b-\frac{da}{ck},m+b-\frac{da}{ck}})\notag\\
    =&\sum_{r \in (\mathbb{Z}/k\mathbb{Z})^{2N}}\sum_{g \in (\mathbb{Z}/d\mathbb{Z})^{2N}} \exp(2\pi i \frac{c}{d} k\braket{g+dr+b-\frac{da}{ck},g+dr+b-\frac{da}{ck}})\notag\\
    =&\sum_{r \in (\mathbb{Z}/k\mathbb{Z})^{2N}}\exp(-2\pi i d\braket{a,2r})\sum_{g \in (\mathbb{Z}/d\mathbb{Z})^{2N}} \exp(2\pi i \frac{c}{d} k\braket{g+b-\frac{da}{ck},g+b-\frac{da}{ck}})\notag
\end{align}
In the second line we used the fact that $\braket{a,a}=\braket{b,b}=\braket{a,b}=0$.
Hence, we can do the two sums separately,
\begin{align}
    &\sum_{r \in (\mathbb{Z}/k\mathbb{Z})^{2N}}\exp(-2\pi i d\braket{a,2r})\sum_{g \in (\mathbb{Z}/d\mathbb{Z})^{2N}} \exp(2\pi i \frac{c}{d} k\braket{g+b-\frac{da}{ck},g+b-\frac{da}{ck}})\\
    =&k^{2N}\delta(da^i\in\mathbb{Z}^N)\sum_{g \in (\mathbb{Z}/d\mathbb{Z})^{2N}} \exp(2\pi i\frac{c}{d} k\braket{g+b-\frac{da}{ck},g+b-\frac{da}{ck}})\notag\\ 
    =&k^{2N}\delta(da^i\in\mathbb{Z}^N)d^N (kc,d)^N\delta(ckb^i-da^i\text{ mod }(kc,d))\notag\\
    =&k^{2N}d^N (k,d)^N\delta(ckb^i-da^i\text{ mod }(k,d))\notag
\end{align}
in which we have omitted the exponential sum that's exactly the same as the last section. In the last equality the two delta functions merge: if $da^i\not\in \mathbb{Z}^N$, then the last delta function can never be satisfied.\par
We can thus plug this result back into the average:
\begin{equation}
    \braket{\Theta_{T}(a^\mu, b^\mu, \tau)}=\gamma_b+\frac{1}{k^{N}} \sum_{(c,d)=1,d>0} d^{-N}(k,d)^N \left|\frac{\tau}{k}-\frac{c}{d}\right|^{-N} \delta(ckb^i-da^i\text{ mod }(k,d))
\end{equation}
Similarly we write $k=(k,d)p,d=(k,d)q,(p,q)=1$. Thus
\begin{align}\label{avg}
    \braket{\Theta_{T}(a^\mu, b^\mu, \tau)}&=\gamma_b+ \sum_{(pc,q)=1,q>0} |q \tau-p c|^{-N}\delta(pckb^i-kqb^i\text{ mod }k)\\
    &=\gamma_b+ \sum_{(c,d)=1,c>0} |c\tau-d|^{-N}\delta(cka^i-dkb^i\text{ mod }k) \notag\\
    &=\gamma_b+ \sum_{(c,d)=1,c>0} |c\tau-d|^{-N}\delta(ca^i-db^i\in\mathbb{Z}^N)\notag
\end{align}
which is the final expression for the averaged classical partition function of the twisted boson, Eq. (\ref{specialAverage})\par

\subsection{Derivation from modular differential equation}

For a self-contained derivation of our general averaging formula \eqref{generalAverage}, we can also apply the method of Maloney and Witten \cite{2006.04855}. It is straightforward to check by a direct calculation that $\Theta_{H,T}$ in \eqref{fHS} obeys the differential equation 
\be\label{mweq}
(\Delta_\tau - N \frac{\p}{\p \tau_2} - \Delta_M) \Theta_{H,T}(A^\mu, \tau) = 0
\ee
where $\Delta_{M}$ is the Laplacian on moduli space acting on $H_{\mu\nu}$ (or $G_{ij}$ and $B_{ij}$), and $\Delta_{\tau}$ is the hyperbolic Laplacian on the upper-half plane,
$\Delta_{\tau} =  - \tau_2^2\left( \frac{\p^2}{\p \tau_1^2} + \frac{\p^2}{\p \tau_2^2} \right) $ .
Inserting this differential equation into the averaging integral \eqref{averageint} leads to the eigenvalue equation 
\be\label{evalue}
(\Delta_{\tau} + \frac{N}{2}(\frac{N}{2}-1)) \tau_2^{N/2} \big\langle \Theta_T(A^\mu, \tau) \big\rangle = 0 \ .
\ee
So far, this is exactly the argument of Maloney and Witten, who considered the case with trivial twists $A^\mu=0$ \cite{2006.04855}. All the twisted sectors obey the same differential equation, but the difference is that with nontrivial twists, $\Theta$ is no longer invariant under $SL(2,\mathbb{Z})$; it transforms in some finite-dimensional representation. In fact the decomposition \eqref{thetaresidues} and the results of \cite{Schoeneberg,Vigneras} on the $a=0$ theta functions tell us that $\Theta_{H,T}$ is a modular form for the congruence subgroup $\Gamma(k^2)$.

It is easy to check that the Eisenstein series \eqref{generalAverage} satisfies the eigenvalue equation. We can also check that \eqref{generalAverage} transforms in the correct representation, as follows. First we rewrite the average as 
\begin{align}
\big\langle \Theta_T(A^\mu, \tau) \big\rangle=  (\mbox{Im\ }\tau)^{-N/2} \sum_{\gamma \in \Gamma_{\infty} \backslash SL(2, \mathbb{Z})} (\mbox{Im\ } \gamma \tau)^{N/2} \Theta_{H,T}(\gamma A^\mu, i\infty)  \ .
\end{align}
Now act on this with $\sigma \in SL(2,\mathbb{Z})$ by replacing $A^\mu \to \sigma A^\mu, \tau \to \sigma \tau$. By relabeling $\gamma \to \gamma \sigma^{-1}$ in the resulting sum, we find
\begin{align}
\big\langle \Theta_T(\sigma A^\mu, \sigma \tau) \big\rangle &= (\mbox{Im\ } \sigma \tau)^{-N/2} (\mbox{Im\ }\tau)^{N/2} \big\langle \Theta_T(A^\mu,\tau) \big\rangle
\end{align}
which is indeed the correct transformation \eqref{modtheta}.

So our $\langle \Theta_T \rangle$ has all the right properties to be the average: it obeys the differential equation \eqref{evalue}, has the correct singular behavior, and transforms properly. In fact it is the unique object with these properties. The proof of uniqueness uses the fact that $\langle \Theta_T \rangle$ is a modular form for $\Gamma(k^2)$ and follows exactly the argument in \cite{2006.04855, 2104.14710} so we omit it. 

\section{Fourier transforms}\label{app:fourier}
In this appendix we calculate the Fourier transforms used in section \ref{ss:spectrum} and discuss some additional properties of the $\Gamma_0(k)$ Eisenstein series.

\subsection{Fourier transform of $f_0(\tau)$, $f(\tau)$, and $g(\tau)$}
As explained in the main text, the functions $f_0, f, g$ defined in \eqref{f0fg} are related to the principal Eisenstein series for $\Gamma_0(k)$. For non-prime $k$, the Fourier transform can be found using the methods in, e.g. \cite{goldfeld,MR1474964}, but in this section we assume $k$ is prime, and in this case things are considerably simpler. First let us review the situation for the Eisenstein series for $SL(2,\mathbb{Z})$. To calculate the Fourier transform of
\begin{equation}\label{mod1} 
 f_0(\tau)=  \big\langle \Theta_{\Lambda}(0, 0, \tau) \big\rangle=1+\sum_{(c,d)=1,c>0} |c \tau+ d|^{-N} \ ,
\end{equation}
the first step is to multiply by $\zeta(N)$, because this converts the sum over coprime $(c,d)$ into a sum over unrestricted integers, 
\begin{align}\label{zetaf0}
\zeta(N)   f_0(\tau)
&= \zeta(N) + \sum_{m,n \in \mathbb{Z}, m>0} |m\tau + n|^{-N} \ .
\end{align}
Now the Fourier integral is straightforward to calculate, with the result
\begin{align}
\zeta(N)    f_0(\tau) &=
\zeta(N)+2^{2-N}\pi  \tau_2^{1-N} \zeta(N-1) \frac{\Gamma(N-1)}{\Gamma(N/2)^2}\\&+\frac{2\pi^{N/2}}{\Gamma(N/2)}\sum_{n \neq 0}\sigma_{1-N}(n)  \tau_2^{1/2-N/2} |n|^{N/2-1/2} K_{N/2-1/2}(2\pi |n|  \tau_2) \exp(2\pi i n \tau_1) \notag
\end{align}
where $\sigma_{s}(n) = \sum_{d|n}d^s$ is the divisor function. 
To repeat this for an Eisenstein series defined on the congruence subgroups, we first multiply by a Dirichlet $L$-function. This is the analogue of $\zeta(N)$ for the ordinary Eisenstein series.  
 The Dirichlet $L$-function for a multiplicative character $\chi$ is defined as
\begin{equation}
    L(N,\chi)=\sum_{m\in \mathbb{Z},m>0}\frac{\chi(m)}{m^N} \ .
\end{equation}
For the principal character,
\be
L(N,\chi_0) = \zeta(N)(1-k^{-N}) \ .
\ee
Now consider
\begin{equation}
    \begin{aligned}
f(\tau) =  \big\langle \Theta_{\Lambda}(w(\alpha), 0, \tau) \big\rangle &=1+\sum_{(c,d)=1,c>0} |c \tau+ d|^{-N} \delta(c \alpha \in k Q)\\
&= (\mbox{Im\ }\tau)^{-N/2} E_k(\tau, \frac{N}{2}) \ .
\end{aligned}
\end{equation}
We first multiply by the $L$-function,
\begin{equation}
    \begin{aligned}
        L(N,\chi_0) f(\tau)&=L(N,\chi_0)+\sum_{m,n\in\mathbb{Z},m> 0} \chi_0(n) \frac{1}{|k m \tau+n|^N}
\end{aligned}
\end{equation}
The Fourier transform is now straightforward, for any $k$. For prime $k$, we can also write the sum as
\begin{align}
\sum_{m,n\in\mathbb{Z},m> 0} \chi_0(n) \frac{1}{|k m \tau+n|^N} &=  \sum_{m,n\in\mathbb{Z},m> 0} \frac{1}{|k m \tau+n|^N} - \frac{1}{k^N}  \sum_{m,n\in\mathbb{Z},m> 0} \frac{1}{|m \tau+n|^N} 
\end{align}
Comparing to \eqref{zetaf0}, we recognize this as the difference of two ordinary $SL(2,\mathbb{Z})$ Eisenstein series multiplied by $\zeta(N)$. Thus
\begin{align}\label{ff0rel}
f(\tau)&=\frac{1}{k^N-1} \left(k^N f_0(k \tau) -  f_0(\tau)\right) \ .
\end{align}
The Fourier transform is therefore
\begin{align}
 L(N,\chi_0) f(\tau) =& L(N,\chi_0)+\frac{2^{2-N}\pi  \Gamma(N-1) (k-1)}{k^N \Gamma(N/2)^2} \tau_2^{1-N}\zeta(N-1)\\&+\frac{2\pi^{N/2}\tau_2^{1/2-N/2}}{k^N \Gamma(N/2)} 
        \Big(-\sum_{n \neq 0} \sigma_{1-N}(n) |n|^{N/2-1/2} K_{N/2-1/2}(2\pi |n|  \tau_2) e^{2\pi i n \tau_1} \notag\\
        &\qquad+k^{N/2+1/2} \sum_{n \neq 0} \sigma_{1-N}(n) |n|^{N/2-1/2} K_{N/2-1/2}(2\pi |n| k \tau_2) e^{2\pi i k n \tau_1} \Big)\notag
\end{align}
Now  consider the other twist,
\begin{align}
g(\tau) =  \big\langle \Theta_{\Lambda}(0, w(\beta), \tau) \big\rangle  \ .
\end{align}
By a modular transformation, 
\begin{align}
g(\tau) &= |\tau|^{-N} f(-\frac{1}{\tau})
\end{align}
Using \eqref{ff0rel} and $f_0(-\frac{1}{\tau}) = |\tau|^N f_0(\tau)$ we can express this in terms of the $SL(2,\mathbb{Z})$ Eisenstein series as
\begin{align}
g(\tau) &= \frac{1}{k^N-1}\left( f_0(\frac{\tau}{k}) - f_0(\tau) \right)
\end{align}
Define the twisted divisor function
\begin{equation}
    \sigma_{s}(n,\chi_0)=\sum_{m|n,m>0} m^{s}\chi_0(m) \ .
\end{equation}
The Fourier transform is
\begin{equation}
    g(\tau)=\sum_{l=-\infty}^{\infty} g_l(\tau_2) e^{2\pi i l \frac{\tau_1}{k}}
\end{equation}
where
\begin{equation}\label{goodg0}
    g_0(\tau_2)=\frac{ 2^{2-N}\pi \Gamma(N-1)\tau_2^{1-N}}{k \Gamma(N/2)^2}\frac{L(N-1,\chi_0)}{L(N,\chi_0)}
\end{equation}
and
\begin{equation}\label{goodgl}
    g_l(\tau_2)=\frac{2\pi^{N/2} \tau_2^{1/2-N/2}}{k^{N/2+1/2} \Gamma(N/2)L(N,\chi_0)}|l|^{N/2-1/2} \sigma_{1-N}(l,\chi_0) K_{N/2-1/2}\left(2\pi \frac{|l|\tau_2}{k}\right)
\end{equation}

\subsection{WZW coefficients $h_1$ and $h_2$}

Now we are ready to calculate the WZW coefficients $h_1$ and $h_2$ defined in \eqref{h1}-\eqref{h2},
\begin{equation}
    h_1(\tau)=f_0( k \tau), \qquad h_2(\tau)=\frac{1}{k^N-1} \big( f_0(\tau)-f_0(k \tau) \big) \ .
\end{equation}
%
%
The Fourier expansions are
\begin{align}\label{goodh1}
        h_1(\tau) 
        &=1+\frac{ 2^{2-N} \pi\tau_2^{1-N}}{k^{N-1}} \frac{\zeta(N-1)}{\zeta(N)}\frac{\Gamma(N-1)}{\Gamma(N/2)^2} \\&+\frac{2\pi^{N/2}\tau_2^{1/2-N/2}}{k^N \Gamma(N/2)\zeta(N)} k^{N/2+1/2} \sum_{n \neq 0} \sigma_{1-N}(n) |n|^{N/2-1/2} K_{N/2-1/2}(2\pi |n| k \tau_2) e^{2\pi i k n \tau_1}\notag 
\end{align}
and
\begin{align}\label{goodh2}
        h_2(\tau) 
        =&\frac{ 2^{2-N}\pi \Gamma(N-1)\tau_2^{1-N}}{k^N\Gamma(N/2)^2} \frac{\zeta(N-1)}{\zeta(N)}\left(1-\frac{k-1}{k^N-1}\right)\\
        &+\frac{2\pi^{N/2}\tau_2^{1/2-N/2}}{(k^N-1) \Gamma(N/2)\zeta(N)}\left(\sum_{n \neq 0,k\nmid n} \sigma_{1-N}(n) |n|^{N/2-1/2} K_{N/2-1/2}(2\pi |n| k \tau_2) e^{2\pi i k n \tau_1}\right. \notag \\
        &\left.+\sum_{n \neq 0} \left(\sigma_{N-1}(k n)-\sigma_{N-1}(n)\right) |kn|^{1/2-N/2}K_{N/2-1/2}(2\pi |n| k \tau_2) e^{2\pi i k n \tau_1}\right)\notag
\end{align}
We have rearranged terms in this expression so that the Fourier coefficients are manifestly positive.

\renewcommand{\baselinestretch}{1}\small
\bibliographystyle{ourbst}
\bibliography{merged_draft}
\end{document}